\documentclass[11pt, a4paper]{article}
\usepackage{footnote}
\usepackage[top=1in,left=1in,right=1in,bottom=1in]{geometry}
\usepackage{complexity}
\usepackage[intlimits]{amsmath}
\usepackage{amsfonts,amssymb,amsthm,bbm, mathtools}
\usepackage[colorlinks=true]{hyperref}
\usepackage{verbatim}
\usepackage{enumerate}
\usepackage[vlined,ruled, linesnumbered]{algorithm2e}
\usepackage{xcolor}
\usepackage{todonotes}
\usepackage{authblk}
\usepackage[capitalise, nameinlink]{cleveref}

\usepackage{thmtools}
\usepackage{thm-restate}

\hypersetup{pdfstartview=FitH}

\newtheorem{theorem}{Theorem}

\newtheorem{fact}[theorem]{Fact}

\theoremstyle{definition}
\newtheorem{definition}{Definition}

\newtheorem{claim}[theorem]{Claim}
\newtheorem{observation}[theorem]{Observation}

\newclass{\alt}{alt}
\newclass{\s}{s}
\newclass{\bs}{bs}
\newclass{\fbs}{fbs}
\newclass{\fC}{fC}
\newclass{\Cert}{C}
\newclass{\EC}{EC}
\newclass{\SB}{sb}
\newclass{\fsb}{fsb}
\newclass{\sC}{sC}
\newclass{\qa}{q_{adv}}
\newclass{\disc}{disc}
\newclass{\bias}{bias}
\newclass{\snip}{snip}
\newclass{\Snip}{Snip}
\newclass{\codim}{codim}
\newclass{\I}{I}
\newclass{\Hen}{H}
\newclass{\Div}{D}
\newclass{\nq}{NOQUERY}
\newclass{\q}{QUERY}
\newclass{\lft}{left}
\newclass{\full}{FULL}

\newcommand\supp{\mathop{\rm supp}}

\newcommand{\cA}{\mathcal{A}}
\newcommand{\clB}{\mathcal{B}}
\newcommand{\clC}{\mathcal{C}}
\newcommand{\clE}{\mathcal{E}}
\newcommand{\clL}{\mathcal{L}}

\newcommand{\clQ}{\mathcal{Q}}
\newcommand{\clT}{\mathcal{T}}
\newcommand{\clP}{\mathcal{P}}
\newcommand{\clU}{\mathcal{U}}

\newcommand{\clN}{\mathcal{N}}

\newcommand{\N}{\mathsf{N}}
\newcommand{\clD}{\mathcal{D}}
\newcommand{\clH}{\mathcal{H}}
\newcommand{\rH}{\mathsf{H}}

\newcommand{\II}{\mathrm{I}}
\newcommand{\X}{\mathsf{X}}
\newcommand{\Y}{\mathsf{Y}}
\newcommand{\RS}{\mathrm{RS}}
\newcommand{\sep}{\mathrm{sep}}
\newcommand{\sab}{\mathrm{sab}}

\newcommand{\Q}{\mathsf{Q}}
\newcommand{\ccc}{max conflict complexity}

\makeatletter
\newcommand{\myusepackage}[2][]{\@ifpackageloaded{#2}{} %
 {\ifthenelse{\equal{}{#1}} {\usepackage{#2}} {\usepackage[#1]{#2}} }}
\makeatother

\myusepackage{units}    

\newcommand{\mydef}[2]{\def#1{#2}}

\newcommand{\nospell}[1]{#1}  %

\def\lf#1{\mathopen{}\left#1}
\def\rt#1{\right#1\mathclose{}}

\providecommand{\middle}{\big}
\newcommand{\md}{\middle}

\newcommand{\newident}[3][*]{\ifthenelse{\equal{*}{#1}}%
{\newcommand{#2}[1][]{\Ensuremath{\mathit{#3##1}}}}%
{\newcommand{#2}[1][]{\Ensuremath{\mathit{#3}}}}%
}

\newcommand{\newmat}[3][*]{\ifthenelse{\equal{*}{#1}}%
{\newcommand{#2}[1][]{\Ensuremath{#3##1}}}%
{\newcommand{#2}[1][]{\Ensuremath{#3}}}%
}

\newcommand{\providemat}[3][*]{\ifthenelse{\equal{*}{#1}}%
{\providecommand{#2}[1][]{\Ensuremath{#3##1}}}%
{\providecommand{#2}[1][]{\Ensuremath{#3}}}%
}

\newcommand{\newfunction}[2]{%
\newcommand{#1}[2][*]{\ifthenelse{\equal{*}{##1}}%
{\Ensuremath{#2\lf(##2\rt)}}%
{#2(##2)}}%
}

{}
\newcommand{\MyMakeRefMacros}[3]{\newcommand{#1}[2][]
{\ifthenelse{\equal{}{##1}}{#2~\ref{##2}}{#3~\ref{##1} and~\ref{##2}}}}

\newcommand{\MyMakeEqRefMacros}[3]{\newcommand{#1}[2][]
{\ifthenelse{\equal{}{##1}}{#2~\eqref{##2}}{#3~\eqref{##1} and~\eqref{##2}}}}
{}

{} %

\renewcommand{\qedsymbol}{$\blacksquare$}

\newcommand{\prfstart}[1][]{\ifthenelse{\equal{}{#1}}%
{\begin{proof}\renewcommand{\qedsymbol}{$\blacksquare$}}%
{\begin{proof}[\dgProofOf\ #1]%
\renewcommand{\qedsymbol}{$\blacksquare_{\mbox{\it{\scriptsize{#1}}}}$}}%
}
\newcommand{\prfend}[1][*]{%
\ifthenelse{\equal{}{#1}}{\renewcommand{\qedsymbol}{$\blacksquare$}}{}%
\ifthenelse{\equal{*}{#1}}{}%
{\renewcommand{\qedsymbol}{$\blacksquare_{\mbox{\it{\scriptsize{#1}}}}$}}%
\end{proof}\renewcommand{\qedsymbol}{$\blacksquare$}%
}

\newcommand{\sect}[2][]{
\ifthenelse{\equal{*}{#2}}
{\section*}
{\ifthenelse{\equal{}{#1}}
{\section{#2}}
{\section{#2}\label{#1}}
}
}

\MyMakeRefMacros{\chref}{Chapter}{Chapters}

\MyMakeRefMacros{\sref}{Section}{Sections}

\MyMakeRefMacros{\ssref}{Subsection}{Subsections}

\MyMakeRefMacros{\sssref}{Subsection}{Subsections}

\MyMakeRefMacros{\figref}{Figure}{Figures}

\newcommand{\IfMathMode}[2]{\ifmmode{#1}\else{#2}\fi}

\newcommand{\Ensuremath}{\ensuremath}

\newcommand{\fbr}[1]{\IfMathMode%
{#1}{$#1$}}                     %

\newcommand{\fnbr}[1]{\mbox{\fbr{#1}}}  %

\newcommand{\fla}[2][*]{\ifthenelse{\equal{}{#1}}{\fbr{#2}}{\fnbr{#2}}}

\newcommand{\malabel}[1]{\addtocounter{equation}{1}\tag{\theequation}\label{#1}}
\newcommand{\mal}[2][]{%
\ifthenelse{\equal{}{#1}}%
{\begin{align*} #2 \end{align*}}%
{\ifthenelse{\equal{P}{#1}}%
{\begingroup\allowdisplaybreaks\begin{align*} #2%
\end{align*}\endgroup}%
{\begin{align*} \malabel{#1} #2 \end{align*}}%
}%
}

\newcommand{\m}{\mal}

\MyMakeEqRefMacros{\equref}{Equation}{Equations}

\MyMakeEqRefMacros{\expref}{Expression}{Expressions}

\MyMakeEqRefMacros{\inequref}{Inequality}{Inequalities}

\newcommand{\thrcase}[6]%
{\begin{cases} #1 &\txt{#2}\\ #3 &\txt{#4}\\ #5 &\txt{#6}\end{cases}}

\newcommand{\chs}{\genfrac(){0cm}{}}  %

\newcommand{\PRdg}[2][]{\mathop{\mathbf{Pr}}_{#1}\lf[{#2}\rt]}
\newcommand{\PRr}[3][]{\mathop{\mathbf{Pr}}_{#1}\lf[{#2}\vphantom{|_1^1}\md|\vphantom{|_1^1}{#3}\rt]}

\newcommand{\pl}[1][]{\nospell{\ifthenelse{\equal{}{#1}}%
{\mskip-6mu\stackrel{\text-}{}\mskip-4mu\txt{s}}%
{\fla{#1\mskip-6mu\stackrel{\text-}{}\mskip-4mu\txt{s}}}}}

\newcommand{\fr}[3][*]{%
\ifthenelse{\equal{*}{#1}}%
{\frac{#2}{#3}}{}%
\ifthenelse{\equal{/}{#1}}%
{\nicefrac{#2}{#3}}{}%
\ifthenelse{\equal{}{#1}}%
{\lf.#2\md/#3\rt.}{}%
\ifthenelse{\equal{p_}{#1}}%
{\lf.\lf(#2\rt)\md/#3\rt.}{}%
\ifthenelse{\equal{_p}{#1}}%
{\lf.#2\md/\lf(#3\rt)\rt.}{}%
\ifthenelse{\equal{pp}{#1}}%
{\lf.\lf(#2\rt)\md/\lf(#3\rt)\rt.}{}%
}

\newcommand{\dr}{\nicefrac}

\newcommand{\sq}{\sqrt}

\newcommand{\set}[2][]{\ifthenelse{\equal{}{#1}}%
{\Ensuremath{\lf\{#2\rt\}}}%
{\Ensuremath{\lf\{#2\vphantom{|_1^1}\md|\vphantom{|_1^1}#1\rt\}}}}
\newcommand{\sett}[2]{\Ensuremath{\lf\{#1\vphantom{|_1^1}\md|\vphantom{|_1^1}#2\rt\}}}

\newcommand{\Min}[2][]{\ifthenelse{\equal{}{#1}}%
{\Ensuremath{\min\lf\{#2\rt\}}}%
{\Ensuremath{\min\lf\{#2\vphantom{|_1^1}\md|\vphantom{|_1^1}#1\rt\}}}}

\newfunction{\asO}{O}
\newfunction{\asOm}{\Omega}
\newfunction{\asT}{\Theta}

\mydef{\01}{\set{0,1}}

\newcommand{\sz}[2][]{\ifthenelse{\equal{}{#1}}%
{\Ensuremath{\lf|#2\rt|}}%
{\Ensuremath{\lf|#2\rt|_{#1}}}}

\newcommand{\sszz}[2]{\lf|\vphantom{|_1^1}\lf\{#1\md|#2\rt\}\vphantom{|_1^1}\rt|}

\newcommand{\txt}[1]{\textrm{#1}}  %

\newcommand{\tit}[1]{\textit{#1}}  %

\DeclareMathAlphabet{\mathlowcal}{OT1}{pzc}{m}{it}

\newcommand{\fn}[2][]{%
\IfMathMode{}{}%
\ifthenelse{\equal{}{#1}}%
{\footnote{#2}}%
{\footnote{\label{#1}#2}}%
}

\mydef{\l(}{\lf(}
\mydef{\r)}{\rt)}

\newcommand{\tm}{\cdot}
\newcommand{\xor}{\oplus}
\newcommand{\sbseq}{\subseteq}
\mydef{\eps}{\varepsilon}
\newcommand{\deq}{\stackrel{\textrm{def}}{=}}
\newcommand{\unin}{\mathrel{\subset\mkern-13.1mu\sim}}  %

\providemat{\QQ}{\mathbb{Q}}
\providemat{\NN}{\mathbb{N}}
\providemat{\CC}{\mathbb{C}}
\providemat{\RR}{\mathbb{R}}
\providemat{\ZZ}{\mathbb{Z}}

\newcommand{\ds}[1][]
{\ifthenelse{\equal{}{#1}}{\allowbreak\dots}{#1\allowbreak\dots#1}}
\newmat{\dc}{\ds[,]}

{} %

{}  %

\protected \def \dg #1{%
\textcolor{Red}
{
{\normalmarginpar\marginnote{\bl{DG's comment}}}
{\reversemarginpar\marginnote{\bl{DG's comment}}\\}
\IfMathMode{
~~~\txt{#1}~
}{
~\\~~~#1~\\
{\normalmarginpar\marginnote{\bl{\ul{------}}}}
{\reversemarginpar\marginnote{\bl{\ul{------}}}\\}
}
}
\ClassWarning{My Macros}{#1}
}


{} %
\newcommand{\e}{\emph}
{}  %

\newcommand{\bl}[1]{{\bf #1}} %

\newcommand{\bil}[1]{{\bfseries\itshape #1}} %

\providecommand{\ul}[1]{\underline{#1}} %

\newcommand{\tb}{\quad}

\allowdisplaybreaks

\title{A composition theorem for randomized query complexity via max conflict complexity \footnote{This paper is a merger of \cite{GLS18} and \cite{Sanyal18}, together with some new results.}}

\author[1]{Dmitry Gavinsky}
\author[2]{Troy Lee}
\author[3,4,5]{Miklos Santha}
\author[6]{Swagato Sanyal}

\affil[1]{Institute of Mathematics, Czech Academy of Sciences, 115 67 \v Zitna 25, Praha 1, Czech Republic.}
\affil[2]{Centre for Quantum Software and Information, Faculty of Engineering and Information Technology, University of Technology Sydney, Australia\\ \texttt{troyjlee@gmail.com}}
\affil[3]{Centre for Quantum Technologies, National University of Singapore, Block S15, 3 Science Drive 2, Singapore 117543.}
\affil[4]{IRIF, Universit\'e Paris Diderot, CNRS, 75205 Paris, France\\ \texttt{santha@irif.fr}}
\affil[5]{MajuLab, UMI 3654, Singapore}
\affil[6]{Indian Institute of Technology Kharagpur\\ \texttt{swagato@cse.iitkgp.ac.in}}

\begin{document}
\maketitle
\begin{abstract}
Let $\R_\epsilon(\cdot)$ stand for the bounded-error randomized query complexity with error $\epsilon > 0$. For any relation $f \subseteq \{0,1\}^n \times S$ and partial Boolean function 
$g \subseteq \{0,1\}^m \times \{0,1\}$, we show that $\R_{1/3}(f \circ g^n) \in \Omega(\R_{4/9}(f) \cdot \sqrt{\R_{1/3}(g)})$, where $f \circ g^n \subseteq \lf(\{0,1\}^m\rt)^n \times S$ is the composition of $f$ and $g$. 
We give an example of a relation $f$ and partial Boolean function $g$ for which this lower bound is tight. 

We prove our composition theorem by introducing a new complexity measure, the \emph{\ccc} $\bar \chi(g)$ of a partial Boolean function $g$. We show $\bar \chi(g) \in \Omega(\sqrt{\R_{1/3}(g)})$ for any (partial) function $g$ and 
$\R_{1/3}(f \circ g^n) \in \Omega(\R_{4/9}(f) \cdot \bar \chi(g))$; these two bounds imply our composition result.  
We further show that $\bar \chi(g)$ is always at least as large as the \emph{sabotage complexity} of $g$, introduced by Ben-David and Kothari \cite{DBLP:conf/icalp/Ben-DavidK16}.  
\end{abstract}

\thispagestyle{empty}

\newpage
\setcounter{page}{1}

\section{Introduction}
\label{intro}
For Boolean functions $f: \{0,1\}^n \rightarrow \{0,1\}$ and $g: \{0,1\}^m \rightarrow \{0,1\}$, the composed Boolean function $f \circ g^n:\lf(\{0,1\}^m\rt)^n \rightarrow \{0,1\}$ is defined as $f \circ g^n (x_1, \ldots, x_n):=f(g(x_1), \ldots, g(x_n))$. 
A natural question in complexity theory is how the complexity of $f \circ g^n$ relates to the complexities of $f$ and $g$.  We study this question in the context of query complexity.
A query algorithm for a function $h$ queries bits of an input $x$, possibly in an adaptive manner, with the goal of outputting $h(x)$ after making as few 
queries as possible.  The deterministic query complexity $\D(h)$ is the minimum over all deterministic query algorithms that compute $h$ of the worst-case number of queries made over all inputs $x$.  
For a composed function, it is easy to see that $\D(f \circ g^n) \leq \D(f) \cdot \D(g)$, since $f\circ g^n$ can be computed by simulating an optimal query algorithm of $f$ and serving every query of this algorithm by running an 
optimal query algorithm for $g$.

For many other natural measures of complexity as well, the complexity of $f \circ g^n$ is similarly bounded from above by the product of the complexities of $f$ and $g$. 
As examples, this holds for randomized query complexity (up to a log factor), quantum query complexity, exact and approximate polynomial degree \cite{DBLP:conf/stoc/Sherstov12a}.  
Showing \emph{lower bounds} on the complexity of $f \circ g^n$ that match these upper bounds as closely as possible 
is usually much more difficult; such results are often referred to as \emph{composition theorems}. Besides answering a natural and interesting structural question, composition theorems find applications to constructing functions that 
are hard with respect to various measures of complexity, and separating complexity measures.

For deterministic query complexity, it can be shown by an adversary argument that $\D(f \circ g^n) = \D(f) \cdot \D(g)$ \cite{DBLP:journals/cjtcs/Montanaro14, Tal13}, i.e., the query algorithm described above is an optimal query algorithm for $f \circ g^n$.  For bounded-error quantum query complexity $Q(f)$, a perfect composition theorem is also known $\Q(f \circ g^n) = \Theta(\Q(f) \cdot \Q(g))$ \cite{DBLP:conf/stoc/HoyerLS07,DBLP:conf/soda/Reichardt11a}.  
Of the three main variants of query complexity, the one that has thus far resisted a complete understanding of its nature under composition is randomized query complexity.

Recently, Ben-David and Kothari \cite{DBLP:conf/icalp/Ben-DavidK16} made progress on this question by showing that for any partial function $f$ 
and total function $g$,
\begin{equation}
\label{eq:BDK}
R_{1/3}(f \circ g^n) \in \Omega \left(R_{1/3}(f) \cdot \sqrt{\frac{R_0(g)}{\log R_0(g)}}\right) \enspace .
\end{equation}
They did this by introducing a new complexity measure, the \emph{sabotage complexity}, $\RS(g)$ of $g$ and showing that 
$R_{1/3}(f \circ g^n) \in \Omega(R_{1/3}(f) \cdot \RS(g))$, and then showing that $\RS(g)$ is quadratically related to $R_0(g)$ for total functions 
$g$.  For partial functions $g$, unbounded separations are known between $\RS(g)$ and $\R_{1/3}(g)$.  

In this work, we introduce another complexity measure called the \emph{\ccc} $\bar \chi(g)$ of $g$.  We show that $\bar \chi(g)$ is a quadratically 
tight lower bound on randomized query complexity, even for partial functions $g$.
\begin{restatable}{thm}{maina}
\label{max_quad}
For any partial Boolean function $g \subseteq \{0,1\}^m \times \{0,1\}$,
\[\bar \chi(g)\in\Omega\lf(\sqrt{\R_{\dr13}(g)}\rt) \enspace.\]
\end{restatable}
Our main application, and the motivation behind the definition of \ccc, comes from showing a composition theorem.   We use \ccc \ to show a composition theorem in a more general 
setting where $f$ can be a relation and $g$ can be a partial function.
\begin{restatable}[Main Theorem]{thm}{main}
\label{main}
For any relation $f \subseteq \{0,1\}^n \times S$ and partial Boolean function $g \subseteq \{0,1\}^n \times \{0,1\}$,
\[\R_{\dr13}(f \circ g^n) \in \Omega\lf(\R_{\dr49}(f) \cdot \bar \chi(g)\rt).\]
\end{restatable}
Putting \cref{max_quad} and \cref{main} together we have the following corollary.
\begin{restatable}{corollary}{maincor}
\label{maincor}
For any relation $f \subseteq \{0,1\}^n \times S$ and partial Boolean function $g \subseteq \{0,1\}^n \times \{0,1\}$,
\[\R_{\dr13}(f \circ g^n) \in \Omega\lf(\R_{\dr49}(f) \cdot \sqrt{\R_{\dr13}(g)}\rt).\]
\end{restatable}

We further show that $\bar \chi(g)$ is always at least as large as the sabotage complexity of $g$.
\begin{restatable}{thm}{sabotage}
\label{sabotage}
For any partial Boolean function $g \subseteq \{0,1\}^m \times \{0,1\}$,
\[\bar \chi(g) \ge \RS(g) \enspace .\]
\end{restatable}
Thus for a partial function $f$ and total function $g$, \cref{main} also implies the bound of Ben-David and Kothari \cref{eq:BDK}.  

Finally, we show that in the more general setting where $f$ can be a relation and $g$ a partial function, \cref{maincor} 
can be \emph{tight}.
\begin{restatable}{thm}{}
\label{t_match}
There exists a relation $f_0 \subseteq \{0,1\}^n \times \{0,1\}^n$ and partial Boolean function $g_0 \subseteq \01^n \times \01$ such that
\m{
  \R_{\dr49}(f_0)\in\asT{\sq n},\,
  \R_{\dr13}(g_0)\in\asT{n}
  \txt{~and~}
  \R_{\dr13}(f_0\circ g_0^n)\in\asT{n} \enspace.}
\end{restatable}

\subsection{Proof Technique}
\label{idea}
At a high level, the proof of Theorem~\ref{main} follows the structure of the proof by Anshu et al.\ \cite{fstcomp} and Ben-David and Kothari \cite{DBLP:conf/icalp/Ben-DavidK16}. We show that for every probability distribution $\eta$ over the 
input space $\{0,1\}^n$ of $f$, there exists a query algorithm $\mathcal{A}$ that makes $O(\R_{\dr13}(f \circ g^n)/\sqrt{\R_{\dr13}(g)})$ queries in the worst case, and computes $f$ with high probability, $\Pr_{z \sim \eta} [(z,\mathcal{A}(z))\in f] \geq 5/9$. By the minimax principle (Fact~\ref{minmax}) this implies Theorem~\ref{main}.

We do this by using a query algorithm for $f \circ g^n$ to construct a query algorithm for $f$.  We define a sampling procedure that for any $z \in \{0,1\}^n$ samples $x=(x_1, \ldots, x_n)$ such that 
$(z,s) \in f$ if and only if $(x,s) \in f \circ g^n$.  This procedure is defined in terms of $\clQ$, which is a probability distribution over pairs of distributions $(\mu_0,\mu_1)$, where $\mu_0$ is supported on $g^{-1}(0)$ and $\mu_1$ is 
supported on $g^{-1}(1)$.  We define a distribution $\gamma_\eta$ over $(\{0,1\}^m)^n$ in terms of this sampling process as follows:
\begin{enumerate}
\item Sample $z=(z_1, \ldots,z_n)$ from $\{0,1\}^n$ according to $\eta$.
\item Independently sample $(\mu_0^{(i)}, \mu_1^{(i)})$ from $\clQ$ for $i=1, \ldots, n$.
\item Sample $x_i=(x^{(1)}_i, \ldots, x^{(m)}_i)$ according to $\mu_{z_i}^{(i)}$ for $i=1,\ldots, n$.  Return $x=(x_1, \ldots, x_n)$.
\end{enumerate}
Notice that steps~(1) and~(2) are independent and the order in which they are performed does not matter.  For future reference, for a fixed $z$ let $\gamma_z(\clQ)$ be the probability distribution defined by the last two steps.  

Now $\gamma_\eta$ is simply a probability distribution over $(\{0,1\}^m)^n$. Thus by the minimax principle (Fact~\ref{minmax} below), there is a deterministic query algorithm $\mathcal{A}'$ of worst-case complexity at most 
$\R_{1/3} (f \circ g^n)$ such that $\Pr_{x \sim \gamma_\eta}[(x,\mathcal{A}'(x)) \in f \circ g^n] \geq 2/3$. We first use $\mathcal{A}'$ to construct a randomized query algorithm $T$ for $f$ with bounded \emph{expected} 
query complexity and error at most $\dr13$.  $T$ is presented formally in Algorithm~\ref{Tee}. The final algorithm $\mathcal{A}$ will be a truncation of $T$ which has bounded worst-case complexity and error at 
most $\dr49$.

On input $z$, the algorithm $T$ seeks to sample a string $x$ from $\gamma_z(\clQ)$, and run $\mathcal{A'}$ on $x$.  Put another way, $\gamma_z(\clQ)$ induces a probability distribution over the 
leaves of $\mathcal{A}'$, and the goal of $T$ is to sample a leaf of $\mathcal{A}'$ according to this distribution.
Since for each $s \in S$, $(x,s) \in f \circ g^n$ if and only if $(z,s) \in f$, and 
$\Pr_{x \sim \gamma_\eta}[(x,\mathcal{A}'(x)) \in f \circ g^n] \geq \dr23$, we have that $\Pr_{z \sim \eta}[(z,T(z)) \in f ] \geq \dr23$. Thus $T$ meets the accuracy requirement.

The catch, of course, is to specify how $T$ samples from $\gamma_z(\clQ)$ \emph{without making too many queries to $z$}.  To sample $x_i$ from $\mu_{z_i}^{(i)}$ seems to require knowledge of $z_i$, and 
thus $T$ would have to query all of $z$.

To bypass this problem, we remember that $\mathcal{A}'$, being an efficient algorithm, will query only a few bits of $x$. This allows us to sample $x$ bit by bit as and when they are queried by $\mathcal{A}'$.
To see this more clearly, consider a run of $T$ where the pairs of distributions $(\mu_0^{(1)}, \mu_1^{(1)}), \ldots, (\mu_0^{(n)}, \mu_1^{(n)})$ were chosen in step~(2) of the sampling procedure.  
Suppose that $T$ is trying to simulate $\mathcal{A'}$ at a vertex $v$ where $x_i^{(j)}$ is queried.  To respond to this query, $T$ will sample $x_i^{(j)}$ from its marginal distribution according to $\mu_{z_i}^{(i)}$ conditioned on the 
event $x \in v$. Let the following be the marginal distributions of $x_i^{(j)}$ for the two possible values of $z_i$.
\[\begin{array}{c|c|c}
& \Pr_{x_i \sim \mu_{z_i}^{(i)}} [x_i^{(j)}=0 \mid x \in v]&\Pr_{x_i \sim \mu_{z_i}^{(i)}} [x_i^{(j)}=1 \mid x \in v] \\
\hline \hline
z_i=0 & p_0 & 1-p_0 \\
\hline
z_i=1& p_1 & 1-p_1 \\
\hline
\end{array}
\]

Without loss of generality, assume that $p_0 \leq p_1$. $T$ answers the query by the procedure \textsc{Bitsampler} given in Algorithm~\ref{samplepr}.
\begin{algorithm}[!h]\label{samplepr}
\DontPrintSemicolon
\caption{ \textsc{Bitsampler} (suppose $p_0 \leq p_1$)}
Sample $r \sim [0,1]$ uniformly at random.\label{pickr} \;
\If{$r < p_0$}{ return $0$.} \;
\ElseIf{$r > p_1$}{ return $1$.}\;
\Else{query $z_i$.\label{querystep} \;
\If{$r \leq p_{z_i}$} {return $0$.\label{querystep1}} \Else {return $1$. \label{querystep2}}}
\end{algorithm}

Note that the bit returned by \textsc{Bitsampler} has the desired distribution. The step in which \textsc{Bitsampler} returns the bit depends on the value of $r$ sampled in step~\ref{pickr}. In particular, $z_i$ is queried if and only if $r \in [p_0,p_1]$, and the bit is returned in step~\ref{querystep1} or~\ref{querystep2}. Such a query to $z_i$ contributes to the query complexity of $T$. Thus the probability that $T$ makes a query when the underlying simulation of $\mathcal{A}'$ is at vertex $v$ is $(p_1-p_0)$. We refer to this quantity as $\Delta(v)$. It plays an important role in our analysis (see Section~\ref{ccnr} and Appendix~\ref{key-inf}).

Our sampling procedure and the tools we use to bound its cost is reminiscent of work of Barak et al. \cite{DBLP:journals/siamcomp/BarakBCR13} in communication complexity.  
They look at a communication analog of our setting where two players are trying to sample a leaf in a communication protocol while communicating as little as possible.  

\subsubsection{Conflict complexity and \ccc}
Bounding the query complexity of $T$ naturally suggests the quantities that we define in this work: the conflict complexity $\chi(g)$ and the \ccc \ $\bar \chi(g)$ of a partial Boolean function $g$. A formal definition can be found in Section~\ref{cc}; here we give the high-level idea and motivation behind these quantities.  

Forget about $T$ for a moment and just consider a deterministic query algorithm $\clB$ computing the partial function $g \subseteq \{0,1\}^m \times \{0,1\}$.  Let $\mu_0, \mu_1$ be distributions with support on 
$g^{-1}(0), g^{-1}(1)$, respectively.  For each vertex $v \in \clB$ let $p_0(v)$ (respectively $p_1(v))$ be the probability that the answer to the query at $v$ is 0 on input $x \sim \mu_0$ (respectively $x \sim \mu_1$), conditioned 
on $x$ reaching $v$.  Now we can imagine a process $\clP(\clB, \mu_0, \mu_1)$ that runs BITSAMPLER on the tree $\clB$: $\clP(\clB, \mu_0, \mu_1)$ begins at the root, and at a vertex $v$ in $\clB$ it uniformly chooses a random 
real number $r \in [0,1]$.  
If $r < \min\{p_0(v), p_1(v)\}$ then the query is ``answered'' 0 and it moves to the left child.  If $r > \max\{p_0(v), p_1(v)\}$ then the query is ``answered'' 1 and it moves to the right child.  If 
$r \in [\min\{p_0(v), p_1(v)\}, \max\{p_0(v), p_1(v)\}]$ then the process halts.  The conflict complexity $\chi(\clB, (\mu_0, \mu_1))$ is the expected number of vertices this process visits before halting.  The conflict complexity of $g$ 
is defined to be 
\[
\chi(g) = \max_{(\mu_0, \mu_1)} \min_{T} \chi(T, (\mu_0, \mu_1)) \enspace,
\]
where the minimum is taken over trees $T$ that compute $g$.  
For \emph{\ccc} \ we enlarge the set over which we maximize.  Let $\clQ$ be a distribution over pairs of distributions $(\mu_0, \mu_1)$, where $\supp(\mu_0) \subseteq g^{-1}(0), \supp(\mu_1) \subseteq g^{-1}(1)$ 
for each pair $(\mu_0, \mu_1)$ in the support of $\clQ$.
Let $\chi(\clB, \clQ) = \E_{(\mu_0, \mu_1) \sim \clQ}\ [ \chi(\clB, (\mu_0, \mu_1))]$.
The max conflict complexity $\bar \chi(g)$ is defined as 
\[
\bar \chi(g) = \max_{\clQ} \min_{T} \chi(T, \clQ) \enspace ,
\]
where the minimum is taken over trees $T$ that compute $g$.  
Clearly, the \ccc \ is at least as large as the conflict complexity.

To motivate the \ccc, note that the query complexity of $T$ is the number of times step~\ref{querystep} in \textsc{Bitsampler} is executed, i.e.\  when the random number $r \in [p_0, p_1]$.  In the definition of $T$ we will choose 
$\clQ$ to achieve the optimal value in the definition of $\bar \chi(g)$.  Then intuitively one expects that for each $i$, $T$ queries $z_i$ only after $\mathcal{A}'$ makes about $\bar \chi(g)$ queries into $x_i$.     
By means of a direct sum theorem for \ccc \ we make this intuition rigorous and prove that the expected query complexity of $T$ is at most $\R_{\dr13}(f \circ g^n)/ \bar \chi(g)$. 
We refer the reader to Section~\ref{comp} for a formal proof.

\subsubsection{$\bar \chi(g)$ and $\R(g)$}
Note that applying \cref{main} with the outer function $f(z) = z_1$ shows that $\R_{\dr13}(g) \in \Omega(\bar \chi(g))$.
We complete the proof of \cref{maincor} by showing that \ccc \ is a quadratically tight lower bound on randomized query complexity, even for partial functions $g$.  In fact, in \cref{maina} we show the stronger result that this is true even for the 
conflict complexity $R_{\dr13}(g) \in O(\chi(g)^2)$.

To prove $R(g) \in O(\chi(g)^2)$, we again resort to the minimax principle; we show that for each probability distribution $\mu$ over the valid inputs to $g$, there is an 
accurate and efficient distributional query algorithm for $g$. For $b \in \{0,1\}$, let $\mu_b$ be the distribution obtained by conditioning $\mu$ on the event $g(x)=b$. By the definition of $\chi(g)$, there is a query algorithm $\clB$ 
such that the following is true: if its queries are served by \textsc{Bitsampler}, step~\ref{querystep} is executed within expected $\chi(\clB, \mu_0, \mu_1) \leq \chi(g)$ queries. Note that at a vertex $v$ which queries $i$, the probability that 
step~\ref{querystep} is executed is $\Delta(v) = |\Pr_{\mu_0}[x_{i}=0 \mid x \mbox{ at } v] - \Pr_{\mu_1}[x_{i}=0 \mid x \mbox{ at } v]|$. This roughly implies that for a typical vertex $v$ of $\mathcal{B}$, $\Delta(v)$ is at least about $\frac{1}{\chi(g)}$. By a technical claim that we prove (Claim~\ref{mutin}) this implies that 
the query outcome at $v$ carries about $\frac{1}{\chi(g)^2}$ bits of information about $g(x)$. Using the \emph{chain rule of mutual information}, we can show that the mutual information between $g(x)$ and the outcomes of first 
$O(\chi(g))^2$ queries by $\clB$ is $\Omega(1)$. This enables us to conclude that we can infer the value of $g(x)$ with success probability $1/2 + \Omega(1)$ from the transcript of $\clB$ restricted to the first $O(\chi(g)^2)$ queries. 
The distributional algorithm of $g$ for $\mu$ is simply the algorithm $\clB$ terminated after $O(\chi(g)^2)$ queries.

\subsubsection{$\bar \chi(g)$ \ and $\RS(g)$}
To see why $\bar \chi(g) \ge \RS(g)$, we first give an alternative characterization of $\RS(g)$.  For a deterministic tree $T$ computing $g$ and strings $x,y$ such that $g(x) \ne g(y)$, let $\sep_T(x,y)$ be the depth of the node $v$ in $T$ 
such that $x$ and $y$ both reach $v$ yet $x_{q(v)} \ne y_{q(v)}$, where $q(v)$ is the index queried at $v$.  Let $\clT$ be a zero-error randomized protocol for $g$, i.e.\ $\clT$ is a probability distribution supported on deterministic trees 
that compute $g$.  Then we have (for a proof see \cref{app:sabotage})
\[
\RS(g) = \min_{\clT} \max_{x,y \atop g(x) \ne g(y)} \E_{T \sim \clT} [\sep_T(x,y)] \enspace .
\]
By von Neumann's minimax theorem \cite{vonNeumann}, this is equal to 
\[
\RS(g) = \max_{p} \min_T \E_{(x,y) \sim p} [\sep_T(x,y)] \enspace.
\]
Here, the max is taken over distributions $p$ on pairs $(x,y)$ where $g(x) \ne g(y)$, and the min is taken over deterministic trees $T$ computing $g$.  

We have seen that the definition of $\bar \chi(g)$ is 
\[
\bar \chi(g) = \max_{\clQ} \min_{T} \E_{(\mu_0, \mu_1) \sim \clQ}\ [ \chi(T, (\mu_0, \mu_1))] \enspace ,
\]
where $\clQ$ is a distribution over pairs $(\mu_0, \mu_1)$ and $T$ is a deterministic tree computing $g$.  When $(\mu_0, \mu_1)$ are taken to be singleton distributions, 
i.e.\ $\mu_0$ puts all its weight on a single $x$ with $g(x)=0$, and $\mu_1$ puts all its weight on a single $y$ with $g(y)=1$, it is easy to see that $\chi(T, (\mu_0, \mu_1)) = \sep_T(x,y)$ 
(see \cref{clm:sabotage}).  Thus $\bar \chi(g)$ is at least as large as the sabotage complexity of $g$ as $\clQ$ is allowed to be a distribution over general $(\mu_0, \mu_1)$, not just 
singleton distributions.

\section{Preliminaries}
\label{prelims}
Let $g \subseteq \{0,1\}^m \times \{0,1\}$ be a partial Boolean function. For $b \in \{0,1\}$, $g^{-1}(b)$ is defined to be the set of strings $x$ in $\{0,1\}^m$ for which $(x,b) \in g$ and $(x,\overline{b}) \notin g$. We refer to $g^{-1}(0) \cup g^{-1}(1)$ as the set of valid inputs to $g$. We assume that for all strings $y \notin g^{-1}(0) \cup g^{-1}(1)$, both $(y,0)$ and $(y,1)$ are in $g$. For a string $x \in g^{-1}(0) \cup g^{-1}(1)$, $g(x)$ refers to the unique bit $b$ such that $(x,b) \in g$. All the probability distributions $\mu$ over the domain of a partial Boolean function $g$ in this paper are assumed to have support on $g^{-1}(0) \cup g^{-1}(1)$. Thus $g(x)$ is well-defined for any $x$ in the support of $\mu$.

Let $S$ be any set. Let $h \subseteq \{0,1\}^k \times S$ be any relation. Consider query algorithms $\mathcal A$ that accept a string $x \in\{0,1\}^k$ as input, query various bits of $x$, and produce an element of $S$ as output. We denote the output by $\mathcal{A}(x)$.
\begin{definition}[Deterministic query complexity]
A deterministic query algorithm $\cA$ is said to compute $h$ if $(x, \cA(x)) \in h$ for all $x \in \{0,1\}^k$.
The deterministic query complexity $\D(h)$ of $h$ is the minimum over all deterministic query algorithms $\cA$ computing $h$ of the maximum number of queries made by $\cA$ over $x \in \{0,1\}^k$.  
\end{definition}

\begin{definition}[Bounded-error randomized query complexity]
Let $\epsilon \in [0,1/2)$.  We say that a randomized query algorithm $\cA$ computes $h$ with error $\epsilon$ if $\Pr[(x,\cA(x)) \in h] \geq 1 - \epsilon$ for all $x \in \{0,1\}^k$.
The bounded-error randomized query complexity $\R_\epsilon(h)$ of $h$ is the minimum over all randomized query algorithms $\cA$ computing $h$ with error $\epsilon$ of the 
maximum number of queries made by $\cA$ over all $x \in \{0,1\}^k$ and the internal randomness of $\cA$.
\end{definition}

\begin{definition}[Distributional query complexity]
 Let $\mu$ a distribution on the input space $\{0,1\}^k$ of $h$, and $\epsilon \in [0,1/2)$.  We say that a deterministic query algorithm $\cA$ computes $h$ with distributional error $\epsilon$ on $\mu$ if 
 $\Pr_{x \sim \mu}[(x,\mathcal A(x)) \in h] \geq 1 - \epsilon$.
 The distributional query complexity $\D^\mu_\epsilon(h)$ of $h$ is the minimum over deterministic algorithms $\cA$ computing $h$ with distributional error $\epsilon$ on $\mu$ of the maximum over $x \in \{0,1\}^k$ 
 of the number of queries made by $\cA$ on $x$.
 \end{definition}

We will use the minimax principle in our proofs to go between distributional and randomized query complexity.
\begin{fact}[Minimax principle]
\label{minmax}
For any integer $k >0$, set $S$, and relation $h \subseteq \{0,1\}^k \times \mathcal{S}$,
\[\R_\epsilon(h)=\max_{\mu}\D_\epsilon^\mu(h).\]
\end{fact}
We present a proof of Fact~\ref{minmax} in Appendix~\ref{mmax}.

Let $\mu$ be a probability distribution over $\{0,1\}^k$. We use $\supp(\mu)$ to denote the support of $\mu$.  By $x \sim \mu$ we mean that $x$ is a random string drawn from $\mu$. Let $C \subseteq \{0,1\}^k$ be an arbitrary set 
such that $\Pr_{x \sim \mu}[x \in C]= \sum_{y \in C} \mu(y) > 0$. Then $\mu \mid C$ is defined to be the probability distribution obtained by conditioning $\mu$ on the event that the sampled string belongs to $C$, i.e.,
\[
(\mu \mid C)(x)=\lf\{  \begin{array}{ll} $0$ & \mbox{if $x \notin C$} \\
\frac{\mu(x)}{\sum_{y \in C} \mu(y)} & \mbox{if $x \in C$}\end{array}   \rt.
\]

For a distribution $\clQ$ over pairs of distributions $(\mu_0, \mu_1)$, let $\supp_0(\clQ) = \cup \{ \supp(\mu_0) : \exists \mu_1, (\mu_0, \mu_1) \in \supp(\clQ)\}$.  Similarly let 
$\supp_1(\clQ) = \cup \{ \supp(\mu_1) : \exists \mu_0, (\mu_0, \mu_1) \in \supp(\clQ)\}$.  We say that $\clQ$ is \emph{consistent} if $\supp_0(\clQ)$ and $\supp_1(\clQ)$ are disjoint 
sets.  We say that $\clQ$ is consistent with a (partial) function $g$ if $\supp_0(\clQ) \subseteq g^{-1}(0)$ and $\supp_1(\clQ) \subseteq g^{-1}(1)$.

\begin{definition}[Subcube, co-dimension]
A subset $\C \subseteq \{0,1\}^m$ is called a subcube if there exists a set $S \subseteq \{1, \ldots, m\}$ of indices and an \emph{assignment function} $\sigma:S \rightarrow \{0,1\}$ such that $\C=\{x \in \{0,1\}^m:\forall i \in S, x_i=\sigma(i)\}$. The 
co-dimension $\codim(C)$ of $C$ is defined to be $|S|$. 
\end{definition}
Now we define the composition of a relation and a partial Boolean function.
\begin{definition}[Composition of a relation and a partial Boolean function]
\label{def:comp}
Let $f \subseteq \{0,1\}^n \times S$ and $g \subseteq \{0,1\}^m \times \{0,1\}$ be a relation and a partial Boolean function respectively. The composed relation $f \circ g^n \subseteq \lf(\{0,1\}^m\rt)^n \times S$ is defined as follows: For $x=(x^{(1)}, \ldots, x^{(n)}) \in \lf(\{0,1\}^m\rt)^n$ and $s \in S$, $(x,s) \in f \circ g^n$ if and only if one of the following holds:
\begin{itemize}
\item $x_i \notin g^{-1}(0) \cup g^{-1}(1)$ for some $i \in \{1, \ldots, n\}$.
\item $x_i \in g^{-1}(0) \cup g^{-1}(1)$ for each $i \in \{1, \ldots, n\}$ and $((g(x_1), \ldots, g(x_n)),s) \in f$.
\end{itemize}
\end{definition}
We will often view a deterministic query algorithm as a binary decision tree. In each vertex $v$ of the tree, an input variable is queried. Depending on the outcome of the query, the computation goes to a child of $v$. The child of $v$ corresponding to outcome $b$ of the query is denoted by $v_b$.

The set of inputs that lead the computation of a decision tree to a certain vertex is a subcube. We will use the same symbol (e.g.\ $v$) to refer to a vertex as well as the subcube associated with it.

The execution of a decision tree terminates at some leaf. If the tree computes some relation $h \subseteq \{0,1\}^k \times S$, the leaves are labelled by elements of $S$, and the tree outputs the label of the leaf at which it terminates. We will also consider decision tree with unlabelled leaves (see Section~\ref{cc}).

\section{Conflict Complexity}
\label{sec:conflict}
In this section, we define the conflict complexity and \ccc \ of a partial Boolean function $g$ on $m$ bits.  For this, we will need to introduce some notation related to a deterministic decision tree $T$.  
For a node $v \in T$, let $\pi(v) = \bot$ if $v$ is the root and $\pi(v)$ be the parent of $v$ otherwise.  Let $q(v)$ be the index that 
is queried at $v$ in $T$, and let $d_T(v)$ be the number of vertices on the unique path in $T$ from the root to $v$. The depth of the root is $1$.

Now fix a partial function $g \subseteq \{0,1\}^m \times \{0,1\}$ and probability distributions $\mu_0, \mu_1$ over $g^{-1}(0), g^{-1}(1)$, respectively.  Let $T$ 
be a tree that computes $g$.  For a node $v \in T$ let $p_0(v) =  \Pr_{\mu_0}[x_{q(v)}=0 | x \mbox{ at } v]$ and $p_1(v) = \Pr_{\mu_1}[x_{q(v)}=0 | x \mbox{ at } v]$, and 
\[
R(v) = 
\begin{cases}
1 & \mbox{ if } v \mbox{ is the root} \\
R(\pi(v))\cdot \min\{\Pr_{\mu_0}[x \rightarrow v | x \mbox{ at } \pi(v)], \Pr_{\mu_1}[x \rightarrow v | x \mbox{ at } \pi(v)]\} & \mbox{ otherwise}\enspace.
\end{cases}
\]
Also define
\[
\Delta(v) = |p_0(v) -  p_1(v)| \enspace .
\]
To gather intuition about these quantities, imagine a random walk on $T$ that begins at the root.  At a node $v$, this walk moves to the left child with probability 
$\min\{p_0(v), p_1(v) \}$, and it moves to the right child with probability 
$1- \max\{p_0(v), p_1(v) \}$.  With the remaining probability, $\Delta(v)$, it terminates at $v$.  
Note that for any tree $T$ computing $g$ we have $\sum_{v \in T} \Delta(v) R(v)=1$.  This is because the walk always terminates before it reaches a leaf of $T$.
In particular, this means that $\sum_{v \in T} d_T(v) \Delta(v) R(v)$---the expected number of steps the walk takes before it terminates---is always at most the depth 
of the tree $T$.

\begin{definition}[Conflict complexity and \ccc]
Let $g$ be a partial function.  For distributions $\mu_0,\mu_1$ with $\supp(\mu_b) \subseteq g^{-1}(b)$ for $b \in \{0,1\}$, and a deterministic decision tree $T$ computing 
$g$, define
\[
\chi(T, (\mu_0, \mu_1)) = \sum_{v \in T} d_T(v) \Delta(v) R(v) \enspace .
\]
The conflict complexity of $g$ is 
\[
\chi(g) = \max_{\mu_0,\mu_1} \min_T \chi(T, (\mu_0, \mu_1)) \enspace ,
\]
where the maximum is over all pairs of distributions $(\mu_0,\mu_1)$ supported on $g^{-1}(0)$ and $g^{-1}(1)$ respectively, and 
the minimum is taken over all deterministic trees $T$ computing $g$.
For $\clQ$ a distribution over pairs satisfying $\supp_b(\clQ) \subseteq g^{-1}(b)$ for $b \in \{0,1\}$, and $T$ a 
deterministic tree computing $g$, let 
$\chi(T,\clQ) = \E_{(\mu_0, \mu_1) \sim \clQ} [\chi(T, (\mu_0, \mu_1))]$.  Finally,
the \ccc \ of $g$ is 
\[
\bar \chi(g) = \max_{\clQ} \min_T \chi(T, \clQ) \enspace ,
\]
where the maximum is taken over $\clQ$ with $\supp_b(\clQ) \subseteq g^{-1}(b)$ for $b \in \{0,1\}$, and the minimum is taken over deterministic trees $T$ computing $g$.
\end{definition}

We can extend the definition of conflict complexity and \ccc \ to more general query processes that do not necessarily compute a function.  We first 
need the notion of $\full$.  
\begin{definition}
\label{def:chifull}
For a deterministic tree $T$ and pair of distributions $(\mu_0, \mu_1)$ with disjoint support, we say that $(T, (\mu_0, \mu_1))$ is $\full$ if $\sum_{v \in T} \Delta(v) R(v)=1$, i.e. 
if the random walk described above terminates with probability $1$.  We say that $(T, \clQ)$ is $\full$ if $(T, (\mu_0, \mu_1))$ is $\full$ for each $(\mu_0, \mu_1) \in \supp(\clQ)$.
\end{definition} 

\begin{definition}
For a deterministic tree $T$ and pair of distributions $(\mu_0, \mu_1)$ such that $(T, (\mu_0, \mu_1))$ is $\full$, define 
$\chi(T, (\mu_0, \mu_1))=\sum_{v \in T} d_T(v) \Delta(v) R(v)$.  For a distribution $\clQ$ such that $(T, \clQ)$ is $\full$, define 
$\chi(T,\clQ) = \E_{(\mu_0, \mu_1) \sim \clQ} [\chi(T, (\mu_0, \mu_1))]$.
\end{definition}

\subsection{Comparison with other query measures}
Li \cite{Li18} shows that the conflict complexity of a total Boolean function $g$ is at least the block sensitivity of $g$.   Here we show 
that the \ccc \ of a (partial) function $g$ is at least as large as the sabotage complexity of $g$.  For a total Boolean function $g$, Ben-David 
and Kothari \cite{DBLP:conf/icalp/Ben-DavidK16} show that the sabotage complexity of $g$ is at least as large as the fractional block sensitivity of $g$ \cite{Aar08, Tal13, GSS16}, 
which in turn is at least as large as the block sensitivity.  They also show examples where the sabotage complexity is much larger than the partition bound, quantum query complexity 
and approximate polynomial degree, thus the same holds for \ccc \ as well.  

We first need the following simple claim.  Let $\delta_x$ be the probability distribution that puts weight $1$ on the string $x$.
\begin{claim}
\label{clm:sabotage}
Let $T$ be a deterministic tree computing the partial function $g$ and let $x,y$ be such that $g(x) = 0, g(y) =1$.  Then 
\[
\chi(T, (\delta_x, \delta_y)) = \sep_T(x,y) \enspace .
\]
\end{claim}

\begin{proof}
Let $v_1$ be the root of $T$, and $v_1, v_2, \ldots, v_t$ be the longest sequence of vertices in $T$ that are visited both by $x$ and $y$, i.e. 
$x_{q(v_t)} \ne y_{q(v_t)}$.  For each $i=1, \ldots, t-1$ we see that $\Delta(v_i) = 0$, while $\Delta(v_t)=1$.  Also $R(v_i)=1$ for each $i=1, \ldots, t$, 
while $R(v)=0$ for any other vertex.  Thus $\sum_{v \in T} d(v) \Delta(v) R(v) = t = \sep_T(x,y)$.  
\end{proof}

\begin{theorem}
Let $g \subseteq \{0,1\}^m \times \{0,1\}$ be a partial function.  Then $\bar \chi(g) \ge \RS(g)$.  
\end{theorem}

\begin{proof}
By \cref{thm:alt_sabotage}, 
\[
\RS(g) = \max_p \min_T \E_{(x,y) \sim p} [\sep_T(x,y)] \enspace .
\]
By definition of \ccc \ we have
\[
\bar \chi(g) = \max_{\clQ} \min_T \E_{(\mu_0, \mu_1) \sim \clQ} [\chi(T, (\mu_0, \mu_1))] \enspace .
\]
The distribution $p$ in sabotage complexity is a special case of $\clQ$ where all the pairs of distributions in the support are 
singleton distributions.  The theorem now follows from \cref{clm:sabotage}.
\end{proof}

\section{Query Process}
\label{cc}
We now come to the most important definition of the paper, that of the query process $\clP(\clB, \clQ)$.
Let $t>0$ be any integer and $\clB$ be any deterministic query algorithm that runs on inputs in $\lf(\{0,1\}^m\rt)^t$. Let $x=(x_i^{(j)})_{{i=1, \ldots, t} \atop {j=1, \ldots, m}}$ be a generic input to $\clB$, and let $x_i$ stand for $(x_i^{(j)})_{j=1, \ldots, m}$. 
For a vertex $v$ of $\clB$, let $v^{(i)}$ denote the subcube in $v$ corresponding to $x_i$, i.e., $v=v^{(1)} \times \ldots \times v^{(t)}$. Recall from Section~\ref{prelims} that $v_b$ stands for the child of $v$ corresponding to the query 
outcome being $b$, for $b \in \{0,1\}$.  

The query process $\clP(\clB, \clQ)$ runs on an input $z \in \{0,1\}^t$ and uses the BITSAMPLER (Algorithm~\ref{samplepr}) routine to simulate the queries of $\clB$ to $x$ when it can.  This process is the heart of how we will transform an 
algorithm for $f \circ g^n$ into a query efficient algorithm for $f$. 

\begin{definition}[Query process $\clP(\clB, \clQ)$]
\label{def:P}
Let $\clB$ be a decision tree that runs on inputs $\lf(\{0,1\}^m\rt)^t$.  Let $\clQ$ be a consistent probability distribution over pairs of distributions $(\mu_0, \mu_1)$.  
The query process $\clP(\clB, \clQ)$ is run on an input $z \in \{0,1\}^t$ and is defined by Algorithm~\ref{P}.
\end{definition}

A few comments about \cref{def:P}.  First, we think of $\clB$ and $\clP$ as query procedures that query input variables and terminate. In particular, they do not have to produce outputs, i.e.\ their leaves do not have to be labeled.
Also note that in Algorithm~\ref{P} the segment from line~\ref{sampler} to line~\ref{bitsamplend} corresponds to the \textsc{Bitsampler} procedure in Algorithm~\ref{samplepr}.  Queries to the input bits $z_i$ are made in 
line~\ref{query}, which corresponds to step~\ref{querystep} of \textsc{Bitsampler}.

We now present an important structural result about $\clP (\clB, \clQ)$.  In particular, this formally proves that the procedure \textsc{Bitsampler} given in Algorithm~\ref{samplepr} samples the bits from the right distribution. 
\begin{theorem}[Simulation Theorem]
\label{samedistn}
Let $\clB$ be a deterministic decision tree running on inputs from $(\{0,1\}^m)^t$, and let $v$ be a vertex in $\clB$. 
Let $A_z(v, \clQ)$ be the event that $\clP (\clB, \clQ)$, when run on $z$, reaches node $v$.  
Let $B_z(v, \clQ)$ be the event that for a random input $x$ sampled from $\gamma_z(\clQ)$, the computation of $\clB$ reaches $v$. 
Then for every $z \in \{0,1\}^t$ and each vertex $v$ of $\clB$, 
\[
\Pr[A_z(v,\clQ)] =\Pr[B_z(v,\clQ)] \enspace .
\]
\end{theorem}

The proof of the claim is given in~\cref{same-distn}.
\begin{algorithm}[h]
\label{P}
\DontPrintSemicolon
\caption{ $\clP (\clB, \clQ)$}
\KwIn{$z=(z_1, \ldots, z_t) \in \{0,1\}^t$.}
\For{$1 \leq k \leq t$}
{$\q_k \gets 0$. \;
$\N_k \gets 0$. \;
Sample $(\mu_0^{(k)}, \mu_1^{(k)})$ from $\clQ$.  \label{Qrandom} \;}
$v \gets $Root of $\clB$ \tcp*{Corresponds to $\lf(\{0,1\}^m\rt)^t$}
\While{$v$ is not a leaf of $\clB$}
{
Let $q(v)=(i,j)$, the $j^{th}$ coordinate of $x_i$ \label{P_queries_xi}\;
\If{$\q_i = 0$}
{Sample a fresh real number $r \sim [0,1]$ uniformly at random. \label{sampler}\;
\If{$r < \min_b \Pr_{x_i \sim \mu_b^{(i)}}[x_i^{(j)}=0 \mid x_i \in v^{(i)}]$}
{$v \gets v_0$. \;
}
\ElseIf{$r > \max_b \Pr_{x_i \sim \mu_b^{(i)}}[x_i^{(j)}=0 \mid x_i \in v^{(i)}]$}
{$v \gets v_1$. \;
 }
\Else
{Query $z_i$. $\q_i \gets 1$. \label{query}\;
\If{$r \leq \Pr_{x_i \sim \mu_{z_i}^{(i)}}[x_i^{(j)}=0 \mid x_i \in v^{(i)}]$}
{$v \gets v_0$. \;}
\Else
{$v \gets v_1$. \; \label{bitsamplend}}
}
$ \N_i\gets \N_i+1$. \;}
\Else{
$b \gets \lf\{\begin{array}{ll} $1$ & \mbox{with probability $\Pr_{x_i \sim \mu_{z_i}^{(i)}}[x_i^{(j)}=1 \mid x_i \in v^{(i)}]$} \\ $0$ & \mbox{with probability $\Pr_{x_i \sim \mu_{z_i}^{(i)}}[x_i^{(j)}=0 \mid x_i \in v^{(i)}]$}\end{array}\rt.$ \;
$v \gets v_b$ \;}
}
\end{algorithm}

We will be interested in the number of queries $\clP(\clB,\clQ)$ is able to simulate before making a query to $z_i$.  To this end, 
let the random variable $\clN_i(\clB, z, \clQ)$ stand for the value of the variable $\N_i$ in Algorithm~\ref{P} after the termination of $\clP(\clB, \clQ)$ on input $z$.  Note that $\clN_i$ depends on the 
randomness in the choices of $r$ (\cref{sampler}) and also on the randomness in $\clQ$ in the choice of distributions $(\mu_0^{(k)}, \mu_1^{(k)})$ (\cref{Qrandom}).

\subsection{Relating $\clP(\clB, \clQ)$ to \ccc}
A key to our composition theorem will be relating the number of simulated queries made by $\clP(\clB, \clQ)$ to \ccc, which 
we do in this section.  Let $\clB$ be a query algorithm taking inputs from $\{0,1\}^m$.  In this case, 
$\clN_1(\clB, 1, \clQ) = \clN_1(\clB, 0, \clQ)$.  This is because the behavior of $\clP(\clB, \clQ)$ on input 
$0$ is exactly the same as the behavior on input $1$ before a query to $z$ is made, and after $z$ is queried the value of $\N_i$ does not change.

\begin{claim}
\label{clm:full}
Let $\clB$ be an algorithm taking inputs from $\{0,1\}^m$.  Then
$(\clB, \clQ)$ is $\full$ if and only if $\clP(\clB, \clQ)$ queries $z$ with probability $1$.  If $(\clB, \clQ)$ is $\full$ then 
\[
\chi(\clB, \clQ) = \E[\clN_1(T,1, \clQ)]
\]
\end{claim}

\begin{proof}
Note that until $z$ is queried, $\clP(\clB, (\mu_0,\mu_1))$ exactly executes the random walk described in \cref{sec:conflict}, and querying $z$ in 
$\clP(\clB, (\mu_0,\mu_1))$ corresponds to this random walk terminating.  The first part of the claim then follows as $\clP(\clB, \clQ)$ queries $z$ with probability $1$ 
if and only if $\clP(\clB, (\mu_0, \mu_1))$ queries $z$ with probability $1$ for every $(\mu_0, \mu_1) \in \supp(\clQ)$. 

Also because $\clP(\clB, (\mu_0,\mu_1))$ exactly executes the random walk described in \cref{sec:conflict} we see that 
$\chi(\clB, (\mu_0, \mu_1)) = \E[\clN_1(T,1, (\mu_0, \mu_1))]$.  The second part of the claim follows by taking the expectation of this 
equality over $(\mu_0, \mu_1) \sim \clQ$.
\end{proof}

The correspondence of \cref{clm:full} prompts us to define $\full$ in a more general setting.
\begin{definition}[$\full$]
Let $\clB$ be a query algorithm taking inputs from $(\{0,1\}^m)^t$.
The pair $(\clB, \clQ)$ is said to be $\full$ if for every $z \in \{0,1\}^t$ it holds that $\clP (\clB, \clQ)$ queries $z_i$ with probability $1$, for every $i=1, \ldots, t$.
\end{definition}

\section{Conflict Complexity and Randomized Query Complexity}
\label{ccnr}
In this section, we will prove Theorem~\ref{maina} (restated below).  Our proof relates the conflict complexity to the expected amount of information that is learned about the function value through each query via Pinsker's Inequality.
At a high level, our proof is reminiscent of the result of \cite{DBLP:journals/siamcomp/BarakBCR13} on compressing communication protocols in that both look at a random sampling process to navigate a tree, and relate the probability of 
this process needing to query or communicate at a node to the amount of information that is learned at the node.  
\begin{restatable}{thm}{maina}
\label{maina}
For any partial Boolean function $g \subseteq \{0,1\}^m \times \{0,1\}$,
\[\chi(g)\in\Omega\lf(\sqrt{\R_{\dr13}(g)}\rt).\]
\end{restatable}
\begin{proof}
\label{hibias}
We will show that there exists a constant $\eps < \dr12$ such that for each input distribution $\mu$, $\D_{\eps}^\mu (g)  \leq 10\chi(g)^2$. Theorem~\ref{maina} will follow from the \emph{minimax principle} (Fact~\ref{minmax}), and the observation that the error can be brought down to $\dr13$ by constantly many independent repetitions followed by a selection of the majority of the answers. It is enough to consider distributions $\mu$ supported on valid inputs of $g$. To this end, fix a distribution $\mu$ supported on $g^{-1}(0) \cup g^{-1}(1)$.

Let $\chi(g)=d$. Let $\clB$ be a deterministic query algorithm for inputs in $\{0,1\}^m$ such that $(\clB, \mu_0, \mu_1)$ is $\full$ and $\chi(\mu_0, \mu_1)=\chi(\clB, \mu_0, \mu_1)$. We call such a decision tree an \emph{optimal} decision tree for $\mu_0, \mu_1$. Thus in $\clP(\clB, \mu_0, \mu_1)$, $\E[\clN_1] = \chi(\mu_0, \mu_1) \leq d$. Recall from Section~\ref{cc} that the leaves of $\clB$ can be labelled by bits such that $\clB$ computes $g$ on the supports of $\mu_0$ and $\mu_1$. We assume $\clB$'s leaves to be labelled as such.

Consider the following query algorithm $\clB'$: Start simulating $\clB$. Terminate the simulation if one of the following events occurs. The output in each case is specified below.
\begin{enumerate}
\item If $10d^2$ queries have been made and $v_{10d^2+1}\neq \bot$, terminate and output $\arg \max_b \Pr_{x \sim \mu}[g(x)=b \mid x \in v_{10d^2+1}]$. \label{e2}
\item If $\clB$ terminates, terminate and output what $\clB$ outputs. \label{e1}
\end{enumerate}
By construction, $\clB'$ makes at most $10d^2$ queries in the worst case. The following claim bounds the error of $\clB'$.
\begin{claim}
\label{runtime}
There exists constant $\eps < \dr12$ such that $\Pr_{x \sim \mu}[\clB'(x)\neq g(x)] \leq \eps$. Furthermore, the constant $\eps$ is independent of $\mu$.
\end{claim}
Claim~\ref{runtime} is proven in Appendix~\ref{run-time}. This completes the proof of Theorem~\ref{maina}.
\end{proof}
\section{The Composition Theorem}
\label{comp}
In this section we prove Theorem~\ref{mainb} (restated below).
\begin{restatable}{thm}{mainb}
\label{mainb}
Let $\mathcal{S}$ be an arbitrary set, $f \subseteq \{0,1\}^n \times \mathcal{S}$ be a relation and $g \subseteq \{0,1\}^m \times \{0,1\}$ a partial Boolean function. Then,
\[\R_{\dr13}(f \circ g^n)\in\Omega(\R_{\dr49}(f) \cdot \bar \chi(g)).\]
\end{restatable}
Our proof will make use of the following \emph{direct sum theorem}, which we prove in~\cref{dirprod}.
\begin{restatable}{thm}{}
\label{dp}
Let $\clB$ be a query algorithm acting on inputs from $\lf(\{0,1\}^m\rt)^t$.  Let $\clQ$ be a consistent distribution over pairs of distributions $(\mu_0, \mu_1)$ on $m$-bit strings.  If $(\clB, \clQ)$ is $\full$ then for any 
$z \in \{0,1\}^t$
\[
\sum_{i=1}^t \E[\clN_i(\clB, z,\clQ)] \geq t \cdot \min_{\clC} \chi(\clC, \clQ)] \enspace,
\]
where the minimum is taken over deterministic trees $\clC$ acting on inputs from $\{0,1\}^m$ such that $(\clC, \clQ)$ is $\full$.
\end{restatable}
\begin{proof}[Proof of Theorem~\ref{mainb}]
We shall prove that for each distribution $\eta$ on the inputs to $f$, there is a randomized query algorithm $\cA$ making at most $9\R_{\dr13}(f \circ g^n) /\bar \chi(g)$ queries in the worst case, for which $\Pr_{z \in \eta}[(z,\cA(z)) \in f] \geq \frac{5}{9}$ holds. $\cA$ can be made deterministic with the same complexity and accuracy guarantees by appropriately fixing its randomness. This will imply the theorem by the \emph{minmax principle} (Fact~\ref{minmax}). To this end let us fix a distribution $\eta$ over $\{0,1\}^n$.

Let $\clQ$ be consistent with $g$ such that for any deterministic decision tree $\clC$ computing $g$ we have $\chi(\clC,\clQ) \ge \bar \chi(g)$.
We will use distributions $\eta$ and $\clQ$ to set up a distribution $\gamma_\eta$ over the  input space of $f \circ g^n$.  This distribution is defined as follows:
\begin{enumerate}
\item Sample $z =(z_1, \ldots, z_n)$ from $\eta$.
\item Sample $(\mu_0^{(i)}, \mu_1^{(i)})$ independently from $\clQ$ for $i = 1, \ldots, t$.
\item Sample $x_i$ from $\mu_{z_i}^{(i)}$ for $i=1, \ldots, t$. Return $x = (x_1, \ldots, x_t)$.
\end{enumerate}
Recall from Section~\ref{idea} the observation that for each $z, x$ sampled as above, for each $s \in \mathcal{S}$, $(z,s) \in f$ \emph{if and only if} $(x,s) \in f \circ g^n$.

Assume that $\R_{\dr13}(f \circ g^n)=c$. The minimax principle (Fact~\ref{minmax}) implies that there is a deterministic query algorithm $\cA'$ for inputs from $\lf(\{0,1\}^m\rt)^n$, that makes at most $c$ queries in the worst case, such that $\Pr_{x \in \gamma_\eta}[(x,\cA'(x)) \in f \circ g^n] \geq \frac{2}{3}$. We will first use $\cA'$ to construct a randomized algorithm $T$ for $f$ whose accuracy under the distribution $\eta$ is as desired and which, for every input $z$, 
makes few queries in expectation.
\begin{algorithm}[!h]\label{Tee}
\DontPrintSemicolon
\caption{ $T$}
\KwIn{$z \in \{0,1\}^n$}
Run $\clP(\cA', \clQ)$ on $z$. \;
Return the output of $\cA'$. 
\end{algorithm}
$T$ is described in Algorithm~\ref{Tee}.

First we bound the probability of error by $T$. By~\cref{samedistn}, we have that
$\Pr[(z, T(z)) \in f]= \Pr_{x \sim \gamma_z(\clQ)}[(x,\cA'(x)) \in f \circ g^n]$ for each $z \in \{0,1\}^n$. 
Thus, $\Pr_{z \sim \eta}[(z, T(z)) \in f]=\Pr_{x \sim \gamma_\eta}[(x,\cA'(x)) \in f \circ g^n] \geq \frac{2}{3}$.

Next, we bound the expected number of queries made by $T$ in the worst-case.
\begin{claim}
The expected number of queries made by $T$ on each input $z$ is at most $\frac{c}{\bar \chi(g)}$.
\end{claim}

\begin{proof}
Fix an input $z \in \{0,1\}^n$. For each leaf $\ell=\ell^{(1)} \times \ldots \times \ell^{(n)}$ of $\cA'$ and for each $i=1, \ldots, n$ define $\mathcal{E}_{i, \ell}'$ to be the event that the computation of $\clP(\cA',\clQ)$ finishes at $\ell$ with 
$\mathsf{QUERY_i}=0$. For $i=1, \ldots,n$ define $\mathcal{F}_i'$ to be the event that $\mathsf{QUERY_i}$ is set to $1$ in $\clP(\cA',\clQ)$. Let $\clQ^n$ stand for the (product) distribution of $n$ pairs of probability distributions each independently sampled from $\clQ$. For $i, \ell$ such that $g$ is not constant on $\ell^{(i)}$, let $\clD_{i, \ell}$ be the distribution given by the following sampling procedure:
\begin{enumerate}
\item Sample $(\mu_0^{(1)}, \mu_1^{(1)}), \ldots, (\mu_0^{(n)}, \mu_1^{(n)})$ from $\clQ^n$ conditioned on $\mathcal{E}_{i,\ell}'$,
\item return $(\mu_0^{(i)} \mid \ell^{(i)}, \mu_1^{(i)} \mid \ell^{(i)})$.
\end{enumerate}Let $\clB_{i,\ell}$ be an optimal tree for $\clD_{i,\ell}$, i.e., $\chi(\clB_{i,\ell},\clD_{i,\ell})=\min_{\clC}\chi(\clC, \clD_{i,\ell})$, where the minimization is over all algorithms $\clC$ that output $0$ on 
$\supp_0(\clD_{i,\ell})$ and output $1$ on $\supp_1(\clD_{i,\ell})$. Now, consider the query algorithm $H$ defined in Algorithm~\ref{Qz}.
\begin{algorithm}[!h]
\label{Qz}
\DontPrintSemicolon
\caption{ $H$}
\KwIn{$x \in \lf(\{0,1\}^m\rt)^n$}
Run $\cA'$ on $x$. \;\label{runT}
Let $\cA'$ terminate at leaf $\ell=\ell^{(1)} \times \ldots \times\ell^{(n)}$. \;
\For{$1 \leq i \leq n$}
{\If{$g$ is not constant on $\ell^{(i)}$}
{Run $\clB_{i, \ell}$ on $x_i$. \; \label{runP}}}
\end{algorithm}
Note that $(H, \clQ)$ is $\full$. Now consider a run of the query process $\clP(H, \clQ)$ on input $z$. \cref{dp} implies that $\sum_{i=1}^n \E[\clN_i(H,z, \clQ)] \geq n \cdot \min_{\clC} \chi(\clC,\clQ) =n \cdot \bar \chi(g)$, 
by the choice of $\clQ$. 

Let $\mathcal{F}_i$ to be the event that $\mathsf{QUERY_i}$ is set to $1$ in $\clP(H,\clQ)$ when it reaches a leaf of $\cA'$, and for each leaf $\ell$ of $\cA'$ let $\mathcal{E}_{i, \ell}$ be the event that $\clP(H,\clQ)$ reaches $\ell$ and 
$\mathsf{QUERY_i}=0$ when it does.   Observe that for each $i=1, \ldots, n$, the events $\{\mathcal{F}_i, (\mathcal{E}_{i,\ell})_{\ell}\}$ are mutually exclusive and exhaustive.  

We have that
\begin{align}
n \cdot \bar \chi(g) &\leq \sum_{i=1}^n\E[\clN_i(H,z,\clQ)] \nonumber \\
&=\sum_{i=1}^n\sum_{\ell}\Pr[\clE_{i,\ell}]\cdot \E[\clN_i(H,z,\clQ) \mid \clE_{i,\ell}]+\sum_{i=1}^n\Pr[\mathcal{F}_i]\cdot \E[\clN_i(H,z,\clQ) \mid \mathcal{F}_i] \label{eqn2}
\end{align}
Let $d_i(\ell)$ be the number of queries into $x_i$ made in the unique path from the root of $\cA'$ to $\ell$.
Now, condition on the $n$ pairs of distributions $(\mu_0^{(j)}, \mu_1^{(j)})_{j=1, \ldots, n}$ that are used in $\clP(H,\clQ)$.  We have that,
\begin{align}
\E[\clN_i(H,z,\clQ) \mid \clE_{i,\ell}, (\mu_0^{(j)}, \mu_1^{(j)})_{j=1, \ldots, n}]=d_i(\ell^{(i)})+\E[\clN_1(\clB_{i,\ell},z_i,(\mu_0^{(i)}\mid \ell^{(i)}, \mu_1^{(i)} \mid \ell^{(i)}))].\label{eqn3}
\end{align}
Averaging over $(\mu_0^{(j)}, \mu_1^{(j)})_{j=1, \ldots, n}$ we have from~(\ref{eqn3}) that
\begin{align}
\E[\clN_i(H,z,\clQ) \mid \clE_{i,\ell}] &= d_i(\ell^{(i)})+\E[\clN_1(\clB_{i,\ell},z_i,\clD_{i,\ell})]\nonumber \\
&= d_i(\ell^{(i)})+\min_{\clC}\chi(\clC, \clD_{i,\ell})\mbox{\ \ \ \ (By the choice of $\clB_{i,\ell}$)}. \nonumber \\
&\leq d_i(\ell^{(i)})+\bar \chi(g).\label{eqn4}
\end{align}
Observing that $\sum_{\ell}\Pr[\clE_{i, \ell}]=1-\Pr[\mathcal{F}_i]$, we have from~(\ref{eqn2}) and~(\ref{eqn4}) that
\begin{align}
n\cdot \bar \chi(g) &\leq  \sum_{i=1}^n \sum_{\ell}\Pr[\clE_{i, \ell}] \cdot (d_i(\ell^{(i)})+\bar \chi(g)) + \sum_{i=1}^n\Pr[\mathcal{F}_i]\cdot\E[\clN_i(H,z,\clQ) \mid \mathcal{F}_i]\nonumber \\
&= \sum_{i=1}^n (1-\Pr[\mathcal{F}_i]) \cdot \bar \chi(g)+\sum_{i=1}^n\left(\sum_{\ell}(\Pr[\clE_{i,\ell}]\cdot d_i(\ell^{(i)})+\Pr[\mathcal{F}_i]\cdot \E[\clN_i(H,z,\clQ) \mid \mathcal{F}_i])\right) \nonumber \\
\Rightarrow \sum_{i=1}^n \Pr[\mathcal{F}_i] &\leq  \frac{1}{\bar \chi(g)} \cdot \sum_{i=1}^n\left(\sum_{\ell}(\Pr[\clE_{i,\ell}]\cdot d_i(\ell^{(i)})+\Pr[\mathcal{F}_i]\cdot \E[\clN_i(H,z,\clQ) \mid \mathcal{F}_i])\right).
\label{eqn5}
\end{align}
We will show that $\sum_{i=1}^n\left(\sum_{\ell}(\Pr[\clE_{i,\ell}]\cdot d_i(\ell^{(i)})+\Pr[\mathcal{F}_i]\cdot \E[\clN_i(H,z,\clQ) \mid \mathcal{F}_i])\right) \leq c$. Since $\sum_{i=1}^n \Pr[\mathcal{F}_i]$ is exactly the expected number of queries made by $T$, the claim will follow from~(\ref{eqn5}). 

Consider a run of $\clP(H,\clQ)$ on input $z$, and let $c_i$ be a random variable denoting the number of times \cref{P_queries_xi} of Algorithm~\ref{P} (with $\clB = H$) is a query into $x_i$ before a leaf of $\cA'$ is reached, for $i=1, \ldots, n$.  
Thus $\sum_{i=1}^n \E[c_i] \leq c$.
Further, for each $i,\ell$ we have $d_i(\ell^{(i)}) = \E[c_i \mid \clE_{i,\ell}]$ and $\E[\clN_i(H,z,\clQ) \mid \mathcal{F}_i] =\E[\clN_i(\cA',z,\clQ) \mid \mathcal{F}_i] \leq \E[c_i \mid \mathcal{F}_i]$. Thus,
\begin{align*}
\sum_{i=1}^n\Bigl(\sum_{\ell}(\Pr[\clE_{i,\ell}]\cdot d_i(\ell^{(i)})&+\Pr[\mathcal{F}_i]\cdot \E[\clN_i(H,z,\clQ) \mid \mathcal{F}_i])\Bigr) \\
&\leq \sum_{i=1}^n \Bigl(\sum_{\ell}(\Pr[\clE_{i,\ell}] \cdot \E[c_i \mid \clE_{i,\ell}]) +\Pr[\mathcal{F}_i]\cdot\E[c_i \mid \mathcal{F}_i]\Bigr) \\
&= \sum_{i=1}^n \E[c_i ] \leq c.
\end{align*}
\end{proof}

Now we finish the proof of Theorem~\ref{mainb} by constructing the query algorithm $\cA$. Let $\cA$ be the algorithm obtained by terminating $T$ after $9t/d$ queries. By Markov's inequality, for each $z$, the probability that $T$ makes more than $9t/d$ queries is at most $1/9$. Thus $\cA$ computes $f$ with probability at least $\dr23 - \dr19=\dr59$ on a random input from $\eta$.
\end{proof}

\sect[s_tight]{Tightness: $\R_{\dr13}(f\circ g^n)\in\asO{\R_{\dr49}(f)\tm\sq{\R_{\dr13}(g)}}$ is possible}
In this section we prove Theorem~\ref{t_match}. We construct a relation $f_0\sbseq\01^n\times\01^n$ (i.e., $\mathcal{S}=\01^n$) and a promise function $g_0 \subseteq \01^n\times\set{0,1}$ (i.e., $m=n$), such that $\R_{\dr49}(f_0)\in\asT{\sq n}$, $\R_{\dr13}(g_0)\in\asT{n}$ and $\R_{\dr13}(f_0\circ g_0^n)\in\asT{n}$.

For strings $x=(x_1, \ldots, x_n), z=(z_1, \ldots, x_n)$ in $\{0,1\}^n$, let $x\oplus z$ be the string $(x_1 \oplus z_1, \ldots, x_n \oplus z_n)$ obtained by taking their bitwise XOR. Let $\sz{x}$ stand for the \emph{Hamming weight} $|{\{i \in [n]: x_i =1\}}|$ of $x$. We define $f_0$ as follows:
\m{
  f_0(z)\deq\sett {(a,z) \in \{0,1\}^n \times \{0,1\}^n}{\sz{a\oplus z}\le\fr n2-\sq n}
}
Now we define $g_0$ by specifying $g_0^{-1}(0)$ and $g_0^{-1}(1)$.
\m{
&g_0^{-1}(0)\deq\sett {(x,0)}{x \in \{0,1\}^n, \sz x\le\dr n2-\sq n}, \\
&g_0^{-1}(1)\deq\sett{(x,1)}{x \in \{0,1\}^n, \sz x\ge\dr n2+\sq n}.
}
We now determine the randomized query complexities of $f_0, g_0$ and $f_0 \circ g_0^n$.
\begin{claim}
\label{matchingex}
\begin{enumerate}
\item[(i)] $\R_{\dr49}(f_0)\in\asOm{\sq n}$.
\item[(ii)] $\R_{\dr13}(g_0)\in\asOm n$.
\item[(iii)] $\R_\eps(f_0\circ g_0^n)\in\asO{n\tm\sq{\log(\dr1\eps)}}$.
\end{enumerate}
\end{claim}
We prove Claim~\ref{matchingex} in Appendix~\ref{matchingexample}. Theorem~\ref{t_match} follows from Theorem~\ref{main} and Claim~\ref{matchingex} with $\eps$ set to $\nicefrac 13$.
\paragraph{Acknowledgements.} We thank Rahul Jain for useful discussions. We thank Srijita Kundu and Jevg\={e}nijs Vihrovs for their helpful comments on the manuscript.  T.L.\ would like to 
thank the Simons Institute and the organizers of the ``Workshop on Interactive Complexity'' where part of this work took place.  In particular, we thank Yuval Filmus for 
the suggestion during this workshop to look at the min-max version of conflict complexity, which led to the development of \ccc.

Part of this work was conducted while T.L.\ and S.S.\ were at the Nanyang Technological University and the Centre for Quantum Technologies, supported by the Singapore National Research Foundation under NRF RF Award No.\ 
NRF-NRFF2013-13.  
This work was additionally supported by the Singapore National Research Foundation, the Prime Minister’s Office, Singapore and the Ministry of Education, Singapore under the Research Centres of Excellence programme under research grant R 710-000-012-135.  This research was supported in part by the QuantERA ERA-NET Cofund project QuantAlgo.  D.G.\ is partially funded by the grant P202/12/G061 of GA \v CR and by RVO:\ 67985840.
\begin{appendix}
\section{Minimax principle: proof of Fact~\ref{minmax}}
\label{mmax}
Fix an integer $\ell$. Let $\mathcal{D}_\ell$ be the finite set of all deterministic query algorithms on $k$ bits with worst-case complexity at most $\ell$. Let $\clH_k:=\{0,1\}^k$. For algorithm $\A \in \clD_\ell$ and input $x \in \clH_k$, let $\E(\A,x)=1$ if $(x,\A(x)) \notin h$, and $0$ otherwise. By von Neumann's minimax principle,
\begin{align}\min_{\sigma} \max_{\mu} \sum_{\A \in \clD_\ell, x \in \clH_k} \sigma(\A)\E(\A, x)\mu(x)= \max_{\mu} \min_{\sigma} \sum_{\A \in \clD_\ell, x \in \clH_k} \sigma(\A)\E(\A, x)\mu(x),\label{mm}\end{align}
where $\sigma$ and $\mu$ range over probability distributions over $\clD_\ell$ and $\clH_k$ respectively. Note that in equation~(\ref{mm}), we can assume that the maximum in the left hand side is over  point distributions on $\clH_k$, i.e., distributions that assign weight $1$ to some input $x \in \clH_k$. Similarly we can assume that the minimum in the right hand side is over  point distributions on $\clD_\ell$. Thus we have that,
\begin{align}
\min_{\sigma} \max_{x \in \clH_k} \sum_{\A \in \clD_\ell} \sigma(\A)\E(\A,x)=\max_{\mu} \min_{\A \in \clD_\ell}  \sum_{x \in \clH_k} \E(\A,x) \mu(x).\label{mm1}
\end{align}
From equation~(\ref{mm1}) it follows that
\begin{align*}
\R_\epsilon(h) &= \min \sett{\ell }{ \min_{\sigma} \max_{x \in \clH_k} \sum_{\A \in \clD_\ell} \sigma(\A)\E(\A,x) \leq \epsilon} \\
& \qquad \qquad \qquad \qquad \mbox{(where $\sigma$ ranges over all probability distributions on $\clD_\ell$)} \\
&= \min \sett{\ell} {\max_\mu \min_{\A \in \clD_\ell} \sum_{x \in \clH_k} \E(\A,x) \mu(x) \leq \epsilon}\\
& \qquad \qquad \qquad \qquad \mbox{(where $\mu$ ranges over all probability distributions on $\clH_k$)} \\
&=\max_{\mu} \min \sett{\ell}{\min_{\A \in \clD_\ell} \sum_{x \in \clH_k} \E(\A,x) \mu(x) \leq \epsilon} \\
&=\D^\mu_\epsilon(h).
\end{align*}

\section{Alternative characterization of sabotage complexity}
\label{app:sabotage}
We first go over the standard definition of sabotage complexity from \cite{DBLP:conf/icalp/Ben-DavidK16}.
Let $g \subseteq \{0,1\}^m \times \{0,1\}$ be a partial function.  From $g$, define a partial 
function $g_{\sab} :  P \rightarrow \{\ast,\dagger\}$, where now $P \subseteq \{0,1,\ast, \dagger\}^n$ is defined in the following way.  Let $P^{\ast} \subseteq \{0,1,\ast\}$ 
be the largest set such that for all $z \in P^{\ast}$ there exist $x,y$ with $g(x) \ne g(y)$ and both $x$ and $y$ are consistent with the non-star 
coordinates of $z$.  Define $P^{\dagger} \subseteq \{0,1,\dagger\}$ analogously with $\dagger$ instead of $\ast$.  Then $P = P^{\ast} \cup P^{\dagger}$.  Finally, define $g_{\sab}(z) = \ast$ if 
$z \in P^\ast$ and $g_{\sab}(z) = \dagger$ if $z \in P^{\dagger}$.  
The sabotage complexity of $g$ is defined as $\RS(g) = R_0(g_{\sab})$.  

For a tree $T$ computing $g$, and strings $x,y$ such that $g(x) \ne g(y)$, let $\sep_T(x,y)$ denote the depth of the 
node $v$ in $T$ such that $x$ and $y$ both reach $v$ yet $x_{q(v)} \ne y_{q(v)}$ where $q(v)$ is the index queried at node $v$.  
We have the following alternative characterization of sabotage complexity.

\begin{theorem}
\label{thm:alt_sabotage}
Let $g \subseteq \{0,1\}^m \times \{0,1\}$ be a partial function.  Then 
\begin{align*}
\RS(g) &= \min_\clT \max_{x,y \atop g(x) \ne g(y)} \E_{T \sim \clT} [\sep_T(x,y)] \\
&= \max_p \min_T \E_{(x,y) \sim p} [\sep_T(x,y)]\enspace .
\end{align*}
In the first equation the minimum is taken over zero-error randomized algorithms $\clT$ for $g$.  In the 
second equation, the maximum is taken over distributions over pairs $(x,y)$ where $g(x)=0, g(y)=1$, and 
the minimum is taken over deterministic trees $T$ computing $g$.
\end{theorem}

\begin{proof}
That the right hand side of the first line is equal to the second line follows by von Neumann's minimax theorem \cite{vonNeumann}.

Now we focus on establishing the first line.  We first show that $\RS(g)$ is at most the right hand side of the first line.
Let $\clT^*$ achieve the minimum of the expression on the right hand side.  Let $z \in P$ be any sabotaged input.  Then there 
are $x^*,y^*$ with $g(x^*) \ne g(y^*)$ such that $x^*$ and $y^*$ only differ where $z$ has special symbols.  Thus any 
query that separates $x^*$ and $y^*$ will also find a special symbol.  The expected number of queries to separate 
$x^*$ and $y^*$ is at most $\max_{x,y} \E_{T \sim \clT^*} [\sep_T(x,y)]$, thus the left hand side is at most the right hand side.

For the other direction, let $\clT^*$ be an optimal zero-error randomized algorithm computing $g_{\sab}$.  For any $x,y$ with 
$g(x) \ne g(y)$ we can create $z^\ast \in P^{\ast}$ such that $z^\ast$ has $\ast$ in those positions where $x,y$ disagree, and 
$z^\ast$ agrees with $x,y$ in those positions where they agree with each other..  Let $z^\dag$ equal $z^\ast$
with $\ast$ replaced by $\dag$.  Now $\clT^*$ is able to distinguish between $z^\ast$ and $z^\dagger$ using an expected number 
of queries that is at most $\RS(g)$.  Any query that distinguishes $z^\ast$ and $z^\dagger$ is also a query that separates $x$ and $y$, 
as $z_\ast$ and $z_\dagger$ only differ where $x$ and $y$ do.  This means
\[
\E_{T \sim \clT^*} [\sep_T(x,y)] \le \RS(g) \enspace ,
\]
showing that the right hand side is at most the left hand side.
\end{proof}

\section{Information Theory}
\label{infotheory}
Let $X$ be a random variable supported on a finite set $\{x_1, \ldots, x_s\}$. Let $\mathcal{E}$ be any event in the same probability space. Let $\mathbb{P}[\cdot]$ denote the probability of any event. The \emph{conditional entropy} $\rH(X \mid \mathcal{E})$ of $X$ conditioned on $\mathcal{E}$ is defined as follows.
\begin{definition}[Conditional entropy]
	\[\rH(X \mid \mathcal{E}):=\sum_{i=1}^s \mathbb{P}[X=x_i \mid \mathcal{E}]\log_2 \frac{1}{\mathbb{P}[X=x_i \mid \mathcal{E}]}.\]
\end{definition}
An important special case is when $\mathcal{E}$ is the entire sample space. In that case the above conditional entropy is referred to as the entropy $\rH(X)$ of $X$.
\begin{definition}[Entropy]
	\[\rH(X):=\sum_{i=1}^s\mathbb{P}[X=x_i] \log_2 \frac{1}{\mathbb{P}[X=x_i]}.\]
\end{definition}
Let $Y$ be another random variable in the same probability space as $X$, taking values from a finite set $\{y_1, \ldots, y_t\}$. Then the conditional entropy of $X$ conditioned on $Y$, $\rH(X \mid Y)$, is defined as follows.
\begin{definition}
	\[\rH(X \mid Y)=\sum_{i=1}^t \mathbb{P}[Y=y_i] \cdot \rH(X \mid Y=y_i).\]
\end{definition}
\begin{definition}[Mutual information]
	Let $X$, $Y$ and $Z$ be two random variables in the same probability space, taking values from finite sets. The mutual information between $X$ and $Y$ conditioned on $Z$, $\mathrm{I}(X;Y \mid Z)$, is defined as follows.
	\[\mathrm{I}(X;Y \mid Z):=\rH(X \mid Z)-\rH(X \mid Y, Z).\]
	It can be shown that $\mathrm{I}(X;Y \mid Z)$ is symmetric in $X$ and $Y$: $\mathrm{I}(X;Y \mid Z)=\mathrm{I}(Y;X \mid Z)=\rH(Y \mid Z)-\rH(Y \mid X, Z)$.
\end{definition}
\begin{theorem}[Chain rule of mutual information]
\label{chainrule}
Let $X_1, \ldots, X_k, Y, Z$ be random variables in the same probability space, taking values from finite sets. Then,
\[\II(X_1, \ldots, X_k:Y \mid Z)=\sum_{i=1}^k\II(X_i : Y \mid Z, X_1, \ldots, X_{i-1}).\]
\end{theorem}
\begin{definition}[Kullback-Leibler Divergence]
\label{klpd}
Given two probability distributions $\P$ and $\Q$ on a finite set $\clU$, the \emph{Kullback-Leibler diverrgence} from $\Q$ to $\P$, denoted by $\D(P || Q)$, is defined as:
\[\D(P || Q):=-\sum_{u \in \clU}\P(u) \log \frac{\P(u)}{\Q(u)}.\]
\end{definition}
Given two random variables $\X$ and $\Y$ taking values in a finite set $\clU$, Let $\X \otimes \Y$ be distribution over ordered pairs of elements of $\clU$ (i.e., over elements of $\clU \times \clU$), where the elements are sampled independently according to distributions of $\X$ and $\Y$ respectively. Let $(\X,\Y)$ denote the joint distribution of $\X$ and $\Y$. The following fact can be easily verified.
\begin{fact}
\label{equivalence}
$\I(\X:\Y)=\Div((\X,\Y) || (\X \otimes \Y))$.
\end{fact}
\begin{definition}
\label{l1n}
Given two probability distributions $\P$ and $\Q$ on a finite set $\clU$, the \emph{$\mathsf{L}_1$-distance} between $\P$ and $\Q$, denoted by $||\P-\Q||_1$, is defined as:
\[||\P-\Q||_1:=\sum_{u \in \clU}|\P(u)-\Q(u)|.\]
\end{definition}
\emph{Pinsker's inequality}, stated below, bounds $\D(P || Q)$ in terms of $|\P(u)-\Q(u)|$ from below.
\begin{theorem}[Pinsker's inequality]
\label{pinskers}
Given two probability distributions $\P$ and $\Q$ on a finite set $\clU$,
\[\D(P || Q) \geq \frac{1}{2} ||P-Q||^2_1.\]
\end{theorem}


\section{Proof of~\cref{samedistn}}
\label{same-distn}
In this section we show that $\Pr[A_z(v,\clQ)] =\Pr[B_z(v,\clQ)]$.  To save writing, we fix $z \in \{0,1\}^t$ and $\clQ$ and let $A(v) := A_z(v, \clQ)$ be the 
event that $\P(\clB,Q)$ reaches node $v$ on input $z$, and $B(v):=B_z(v,\clQ)]$ be the event that $\clB$ reaches node $v$ under the distribution
$\gamma_z(\clQ)$.  Additionally, we write $\overline{(\mu_0, \mu_1)} = ((\mu_0^{(1)}, \mu_1^{(1)}), \ldots, (\mu_0^{(t)}, \mu_1^{(t)}))$ for a $t$-tuple of 
pairs of distributions.  In the following when we write $\E_{\overline{(\mu_0, \mu_1)} \sim \clQ^t}$ this expectation is taken with respect to drawing each 
$(\mu_0^{(i)}, \mu_1^{(i)})$ independently from $\clQ$.

Now notice that $\Pr[A(v)] = \E_{\overline{(\mu_0, \mu_1)}\sim \clQ^t} \Pr[A(v) \mid \overline{(\mu_0, \mu_1)}]$ and 
$\Pr[B(v)] = \E_{\overline{(\mu_0, \mu_1)} \sim \clQ^t} \Pr[B(v) \mid \overline{(\mu_0, \mu_1)}]$.  We prove by induction on $d(v)$, the depth of a node $v$, that
\begin{equation}
\label{eq:simul}
\Pr[A(v) \mid \overline{(\mu_0, \mu_1)}]=\Pr[B(v) \mid \overline{(\mu_0, \mu_1)}]
\end{equation}
for any $\overline{(\mu_0, \mu_1)}$.  This will give the claim. 

Towards the aim of showing \cref{eq:simul}, fix an arbitrary $\overline{(\mu_0, \mu_1)}$.

\begin{description}

\item[Base case:] $d(v)=1$, i.e.\ $v$ is the root of $\clB$. Thus $\Pr[A(v) \mid \overline{(\mu_0, \mu_1)}]=\Pr[B(v) \mid \overline{(\mu_0, \mu_1)}]=1$.

\item[Inductive step:] Assume that $d(v) \geq 2$, and that the statement is true for all vertices of depth at most $d(v)-1$. Since $d(v) \geq 2$, $v$ is not the root of $\clB$. Let $u=u^{(1)} \times \ldots \times u^{(t)}$ be the parent of 
$v$, and say variable $x_i^{(j)}$ is queried at $u$. Without loss of generality we assume that $v=u_0$. We split the proof into the following two cases.
\begin{itemize}
\item {\bf Case 1:} $\Pr_{x_i \sim \mu_{z_i}^{(i)}}[x_i^{(j)}=0 \mid x_i \in u^{(i)}] \leq \Pr_{x_i \sim \mu_{\overline{z}_i}^{(i)}}[x_i^{(j)}=0 \mid x_i \in u^{(i)}]$.

Conditioned on $A(u), \overline{(\mu_0,\mu_1)}$ and $\q_i=1$, the probability that $\clP$ reaches $v$ is $\Pr_{x_i \sim \mu_{z_i}^{(i)}}[x_i^{(j)}=0 \mid x_i \in u^{(i)}]$. Also, conditioned on $A(u),\overline{(\mu_0,\mu_1)}$ and 
$\q_i=0$ the probability that $\clP$ reaches $v$ is exactly equal to the probability that the real number $r$ sampled at $u$ lies in $[0, \Pr_{x_i \sim \mu_{z_i}^{(i)}}[x_i^{(j)}=0 \mid x_i \in u^{(i)}] ]$, which is equal to 
$\Pr_{x_i \sim \mu_{z_i}^{(i)}}[x_i^{(j)}=0 \mid x_i \in u^{(i)}]$. Thus,
\begin{align}
\Pr[A(v) \mid \overline{(\mu_0,\mu_1)}]&=\Pr[A(u) \mid \overline{(\mu_0,\mu_1)}] \cdot \Pr[A(v) \mid A(u), \overline{(\mu_0,\mu_1)}] \nonumber \\
&=\Pr[A(u) \mid \overline{(\mu_0,\mu_1)}] \cdot \Pr_{x_i \sim \mu_{z_i}^{(i)}}[x_i^{(j)}=0 \mid x_i \in u^{(i)}]. \label{c1:one}
\end{align}
Now condition on $B(u)$ and $\overline{(\mu_0,\mu_1)}$. The probability that $\clB$ reaches $v$ is exactly equal to the probability that $x_i^{(j)}=0$ when $x$ is sampled according to the distribution 
$\gamma_z(\overline{(\mu_0,\mu_1)})$ conditioned on the event that $x \in u$. Note that in the distribution $\gamma_z(\overline{(\mu_0,\mu_1)})$, the $x_k$'s are independently distributed. Thus,
\begin{align}
\Pr[B(v) \mid \overline{(\mu_0,\mu_1)}]&=\Pr[B(u) \mid \overline{(\mu_0,\mu_1)}] \cdot \Pr[B(v) \mid B(u),\overline{(\mu_0,\mu_1)}] \nonumber \\
&=\Pr[B(u) \mid \overline{(\mu_0,\mu_1)}] \cdot \Pr_{x_i \sim \mu_{z_i}^i}[x_i^{(j)}=0 \mid x_i \in u^{(i)}]. \label{c1:two}
\end{align}
By the inductive hypothesis, $\Pr[A(u) \mid \overline{(\mu_0,\mu_1)}]=\Pr[B(u) \mid \overline{(\mu_0,\mu_1)}]$. It follows from~(\ref{c1:one}) and~(\ref{c1:two}) that 
$\Pr[A(v) \mid \overline{(\mu_0,\mu_1)}]=\Pr[B(v) \mid \overline{(\mu_0,\mu_1)}]$.

\item {\bf Case 2:} $\Pr_{x_i \sim \mu_{z_i}^{(i)}}[x_i^{(j)}=0 \mid x_i \in u^{(i)}] > \Pr_{x_i \sim \mu_{\overline{z}_i^{(i)}}}[x_i^{(j)}=0 \mid x_i \in u^{(i)}]$.
Let $v'=u_1$. By an argument similar to Case 1, we have that 
\begin{align}
\Pr[A(v') \mid \overline{(\mu_0,\mu_1)}]=\Pr[B(v') \overline{(\mu_0,\mu_1)}]. \label{c2}
\end{align}
Now,
\begin{align}
\Pr[A(v) \mid \overline{(\mu_0,\mu_1)}] &=\Pr[A(u) \mid \overline{(\mu_0,\mu_1)}] - \Pr[A(v') \mid \overline{(\mu_0,\mu_1)}] \nonumber \\
&= \Pr[B(u) \mid \overline{(\mu_0,\mu_1)}] - \Pr[A(v') \mid \overline{(\mu_0,\mu_1)}]  \mbox{\ \ \ \ \ (By inductive hypothesis)} \nonumber \\
&= \Pr[B(u) \mid \overline{(\mu_0,\mu_1)}] - \Pr[B(v') \mid \overline{(\mu_0,\mu_1)}] \mbox{\ \ \ \ \ \ (By (\ref{c2}))} \nonumber \\
&= \Pr[B(v) \mid \overline{(\mu_0,\mu_1)}] \enspace. \nonumber
\end{align}
\end{itemize}
\end{description}

\section{Proof of \cref{dp}}
\label{dirprod}
Towards a contradiction, assume that 
\begin{equation}
\sum_{i=1}^t \E[\clN_i(\clB, z,\clQ)] < t \cdot \min_{\clC} \E_{(\mu_0,\mu_1) \sim \clQ} [\chi(\clC, (\mu_0,\mu_1))] \enspace .
\label{hyp}
\end{equation}
By averaging, there exists a $k$ such that $\E[\clN_k(\clB, z, \clQ)]] < \min_{\clC} \E_{(\mu_0,\mu_1) \sim \clQ} [\chi(\clC, (\mu_0,\mu_1))]$. 
Let us focus on the expression on the left hand side.  Recall that there are two kinds of randomness in this expectation, the choice of the 
random numbers $r$ in $\clP(\clB, \clQ)$ and the choice of $\overline{(\mu_0, \mu_1)} \sim \clQ^t$.  We separate out these two as follows:
\begin{align*}
\E[\clN_k(\clB, z, \clQ)] &= \E_{\overline{(\mu_0, \mu_1)} \sim \clQ^t} \E_r[\clN_k(\clB, z, \overline{(\mu_0, \mu_1)}] \\
&= \E_r \E_{\overline{(\mu_0, \mu_1)} \sim \clQ^t} [\clN_k(\clB, z, \overline{(\mu_0, \mu_1)}] \\
&=\E_r \E_{\overline{(\mu_0, \mu_1)}^{-(k)} \sim \clQ^{t-1}} \E_{(\mu_0^{(k)}, \mu_1^{(k)}) \sim \clQ} [\clN_k(\clB, z, \overline{(\mu_0, \mu_1)}] \enspace,
\end{align*}
where $\overline{(\mu_0, \mu_1)}^{-(k)}$ is a $t-1$-tuple of pairs of distributions without the $k^{th}$ coordinate.  This further means that there is a fixing of the randomness $r$ and 
the $(t-1)$-tuple of pairs distributions $\overline{(\mu_0, \mu_1)}^{-(k)}$ such that 
$\E_{(\mu_0^{(k)}, \mu_1^{(k)}) \sim \clQ} [\clN_k(\clB, z, \overline{(\mu_0, \mu_1)}] < \min_{\clC} \E_{(\mu_0,\mu_1) \sim \clQ} [\chi(\clC, (\mu_0,\mu_1))]$.
With such a fixed setting, however,  $\clP(\clB, \clQ)$ creates a query process equivalent to $\clP(\clB',\Q)$ run on $z_i \in \{0,1\}$ for a deterministic query algorithm $\clB'$ running on inputs from $\{0,1\}^m$ and 
such that $(\clB', \clQ)$ is $\full$.  The distribution $\E_{(\mu_0, \mu_1) \sim \clQ}[\clN_1(\clB',1, \mu_0,\mu_1)]$ is the same as that as 
$\E_{(\mu_0^{(k)}, \mu_1^{(k)}) \sim \clQ}[\clN_k(\clB, z, (\overline{(\mu_0, \mu_1)}^{-(k)}, (\mu_0^{(k)}, \mu_1^{(k)}))]$
conditioned on the earlier fixing of $\overline{(\mu_0, \mu_1)}^{-(k)}$ and the randomness $r$.  Thus $\E_{(\mu_0, \mu_1) \sim \clQ} [\chi(\clB', (\mu_0, \mu_1))] < \min_{\clC} \E_{(\mu_0,\mu_1) \sim \clQ} [\chi(\clC, (\mu_0,\mu_1))]$, a contradiction.
\section{Proof of \cref{runtime}}
\label{run-time}
Let $v_k$ be the random vertex at which the $\clB$ makes its $k$-th query when it is run on $x$; If $\clB$ terminates before making $k$ queries, define $v_k:=\bot$. Let $\clE$ denote the event that in at most $10d^2$ queries, the computation of $\clB$ does not reach a vertex $v$ such that $\Pr_{x \sim \mu}[g(x)=0 \mid x \in v]\cdot\Pr_{x \sim \mu}[g(x)=1 \mid x \in v] \leq \frac{1}{9}$. Since $\clB$ computes $g$ on the supports of $\mu_0$ and $\mu_1$, therefore if $\clE$ happens then the computation of $\clB$ does not reach a leaf within $10d^2$ queries. We split the proof into the following two cases.
\begin{description}
\item [Case $1$:] $\Pr[\clE] < \frac{3}{4}$.

Condition on the event that the computation reaches a vertex $v$ of $\clB$ for which $\Pr_{x \sim \mu}[g(x)=0 \mid x \in v]\cdot\Pr_{x \sim \mu}[g(x)=1 \mid x \in v] \leq \frac{1}{9}$ holds. In this case, one of $\Pr_{x \sim \mu}[g(x)=0 \mid x \in v]$ and $\Pr_{x \sim \mu}[g(x)=1 \mid x \in v]$ is at most $\dr13$. Hence,  $|\Pr_{x \sim \mu}[g(x)=0 \mid x \in v]-\Pr_{x \sim \mu}[g(x)=1 \mid x \in v]| \geq \dr13$. Let $w$ be the random leaf of the subtree of $\clB'$ rooted at $v$ at which the computation ends. The probability that $\clB'$ errs is at most
\begin{align*}
&\E_w\lf[\frac{1}{2}-\frac{1}{2}\lf|\Pr_{x \sim \mu}[g(x)=0 \mid x \in w]-\Pr_{x \sim \mu}[g(x)=1 \mid x \in w]\rt|\rt] \\
& \leq \frac{1}{2}-\frac{1}{2} \lf|\E_w\lf[\Pr_{x \sim \mu}[g(x)=0 \mid x \in w]\rt]-\E_w\lf[\Pr_{x \sim \mu}[g(x)=1 \mid x \in w]\rt]\rt| \\
& \qquad \qquad \qquad \qquad \mbox{\ \ \ \ (By Jensen's inequality and linearity of expectation)} \\
&=\frac{1}{2}-\frac{1}{2}\lf|\Pr_{x \sim \mu}[g(x)=0 \mid x \in v]-\Pr_{x \sim \mu}[g(x)=1 \mid x \in v]\rt| \leq \frac{1}{3}.
\end{align*}

Thus we have shown that conditioned on $\overline{\clE}$ the probability that $\clB'$ errs is at most $\frac{1}{3}$. Hence, the probability that $\clB'$ errs is at most $\frac{1}{4}\cdot \frac{1}{3}+\frac{3}{4}\cdot\frac{1}{2} = \frac{11}{24}<\frac{1}{2}$.
\item[Case $2$:] $\Pr[\clE] \geq \frac{3}{4}$.

Let $a_j:=(i_j, x_{i_j})$ be the tuple formed by the index and value of the random input variable queried at the $j$-th step by $\clB'$; if $\clB'$ terminates before making $j$ queries (i.e. $v_j = \bot)$ or $v_j$ is a leaf of $\clB$, then define $i_j, x_{i_j}:=\bot$ . Note that the sequence $(a_1, \ldots, a_{10d^2})$ uniquely specifies a leaf of $\clB'$, and vice versa. Let $\I(\cdot, \cdot)$ denote the \emph{mutual information}. (See Appendix~\ref{infotheory} for the definitions and results from information theory used in this work). We prove the following claim in Appendix~\ref{key-inf}.
\begin{claim}
\label{keyinf}
If $\Pr[\clE] \geq \frac{3}{4}$, then $\I(a_1, \ldots, a_{10d^2}:g(x)) \geq \frac{1}{40}$.
\end{claim}
Thus if $\Pr[\clE] \geq \frac{3}{4}$, Claim~\ref{keyinf} implies that
\begin{align}
\Hen(g(x) \mid a_1, \ldots a_{10d^2}) \leq 1-\frac{1}{40}=\frac{39}{40}. \label{enbound}
\end{align}
Let $\clL$ be the set of leaves $\ell$ of $\clB'$ such that $\Hen(g(x) \mid x \in \ell) \leq \frac{79}{80}$. For each $\ell \in \clL$, $\min_b \Pr_{x \sim \mu}[g(x)=b \mid x \in \ell] \leq \frac{9}{20}$. Conditioned on $(a_1, \ldots, a_{10d^2}) \in \clL$, the probability that $\clB'$ errs is at most $\frac{9}{20}$. By \emph{Markov's inequality} and (\ref{enbound}), it follows that $\Pr[(a_1, \ldots, a_{10d^2}) \in \clL] \geq \frac{1}{79}$. Thus $\clB'$ errs with probability at most $\frac{1}{79}\cdot \frac{9}{20}+\frac{18}{19}\cdot \frac{1}{2}<\frac{1}{2}$.
\end{description}\subsection{Proof of Claim~\ref{keyinf}}
\label{key-inf}
Let $v$ be a vertex in $\clB$. Define $\Delta(v)$ as follows.\footnote{Recall that we mentioned $\Delta(v)$ in Section~\ref{idea}.}
\[\Delta(v):=\lf\{\begin{array}{lll}|\Pr_{x \sim \mu_0} [x_i=0 \mid x \in v]-\Pr_{x \sim \mu_1} [x_i=0 \mid x \in v]| &  \mbox{if $v \neq \bot$ and $\Pr_{x \sim \mu_b} [x \in v] >0$} \\ &\mbox{for $b \in \{0,1\}$,} \\ \ & \  \\ 1 &  \mbox{otherwise.} \end{array}\rt.\]
The following claim shows that if $\Delta(v)$ is large, the query outcome of $v$ contains significant information about $g(x)$.
\begin{claim}
\label{mutin}
Let $v$ be a vertex in $\clB$. Let variable $x_i$ be queried at $v$. Then,
\[\I(g(x) : x_i \mid x \in v)  \geq 8 \lf(\Pr_{x \sim \mu}[g(x)=0 \mid x \in v] \cdot \Pr_{x \sim \mu}[g(x)=1 \mid x \in v] \cdot \Delta(v)\rt)^2.\]
\end{claim}
\begin{proof}
Define $b:=g(x)$. Condition on the event $x \in v$. Recall from Appendix~\ref{infotheory} that $(b \otimes x_i)$ be the distribution over pairs of bits, where the bits are distributed independently according to the distributions of $b$ and $x_i$ respectively. Fact~\ref{equivalence} implies that $\I(b : x_i)=\Div((b,x_i) || (b \otimes x_i))$\footnote{See Appendix~\ref{infotheory} for definition of Kullback-Leibler divergence and $\mathsf{L}_1$-distance.}. Now, \emph{Pinsker's inequality} (Theorem~\ref{pinskers}) implies that 
\begin{align}
\label{fst}\Div((b,x_i) || (b \otimes x_i)) \geq \frac{1}{2} ||(b,x_i)-(b \otimes x_i)||^2_1.
\end{align}
Next, we bound $|(b,x_i)-(b \otimes x_i)||_1$. To this end, we fix bits $z_1, z_2 \in \{0,1\}$, and bound $|\Pr[(b,x_i)=(z_1,z_2)]-\Pr[(b \otimes x_i)=(z_1,z_2)]|$. We have that,
\begin{align}
\label{t1} \Pr[(b,x_i)=(z_1,z_2)]&=\Pr[b=z_1]\Pr[x_i=z_2 \mid b = z_1].
\end{align}
Now,
\begin{align}
\label{t2} \Pr[(b \otimes x_i)=(z_1,z_2)]&=\Pr[b=z_1]\Pr[x_i=z_2] \nonumber \\
&=\Pr[b=z_1](\Pr[b=z_1]\Pr[x_i=z_2 \mid b=z_1]+&\nonumber \\
& \qquad \qquad \qquad \qquad \qquad \Pr[b=\overline{z_1}]\Pr[x_i=z_2 \mid b=\overline{z_1}]).
\end{align}
Taking the absolute difference of~(\ref{t2}) and~(\ref{t1}) we have that,
\begin{align}
&|\Pr[(b,x_i)=(z_1,z_2)]-\Pr[(b \otimes x_i)=(z_1,z_2)]| \nonumber \\
&=\Pr[b=z_1] \cdot \Pr[b=\overline{z_1}] \cdot \Delta(v)=\Pr[b=0] \cdot \Pr[b=1] \cdot \Delta(v)\label{fin}
\end{align}
The Claim follows by adding~(\ref{fin}) over $z_1, z_2$ and using~(\ref{fst}).
\end{proof}
Let  $\clB$ be run on a random input $x$ sampled from $\mu$. The next claim proves a lower bound on the expected sum of $\Delta(v)$ for the random vertices $v$ in the transcript of $\clB$. Recall from Appendix~\ref{run-time} that $v_k$ is the random vertex at which the $k$-th query is made; If $\clB$ terminates before making $k$ queries, define $v_k:=\bot$.
\begin{claim}
\label{sumofdelta2}
Let $c$ be any positive integer. Then,
\[\sum_{k=1}^{10dc} \E[\Delta(v_k) \mid \clE]  \geq \frac{13c}{20}.\]
\end{claim}
To prove Claim~\ref{sumofdelta2} we need the following claim.
\begin{claim}
\label{sumofdelta}
\[\sum_{k=1}^{10d} \E[\Delta(v_k) \mid \clE]  \geq \frac{13}{20}.\]
\end{claim}
Note that if $\clB$ terminates before making $t$ queries, $v_t=\bot$ and $\Delta(v_t)=1$.
\begin{proof}[Proof of Claim~\ref{sumofdelta}]
Let us sample vertices $u_k$ of $\clB$ as follows:
\begin{enumerate}
\item Set $z=\lf\{\begin{array}{ll}
1 & \mbox{with probability $\Pr_{x \sim \mu}[g(x)=1]$}, \\
0 & \mbox{with probability $\Pr_{x \sim \mu}[g(x)=1]$}
\end{array}\rt.$
\item Run $\clP(\clB, \mu_0, \mu_1)$ on the $1$-bit input $z$.
\item Let $u_k$ be the vertex $v$ of $\clB$ in the beginning of the $k$-th iteration of the \emph{while} loop of Algorithm~\ref{P}. If the simulation stops before $k$ iterations, set $u_k:=\bot$. Return $(u_k)_{k=1,\ldots}$.
\end{enumerate}
By \cref{samedistn}, the transcripts $(u_k)_{k=1,2,\ldots}$ and $(v_k)_{{k=1,2,\ldots}}$ have the same distribution.

Now, since $\E[\clN_1] \leq \chi(g) = d$, we have by \emph{Markov's inequality} that the probability that $\clP(\clB, \mu_0, \mu_1)$ sets $\q_1$ to $1$ within first $10d$ iterations of the \emph{while} loop, is at least $9/10$. Note that conditioned on the event that the computation of $\clP(\clB, \mu_0, \mu_1)$ is at vertex $v$ of $\clB$ that queries the input bit $x_i$, the probability that the random real number $r$ generated in the same iteration lies in the interval $[\min_b \Pr_{x_i \sim \mu_b}[x_i=0 \mid x \in v], \max_b \Pr_{x_i \sim \mu_b}[x_i=0 \mid x \in v]]$ is exactly $\Delta(v)$. We have,
\begin{align}
\sum_{k=1}^{10d} \E[\Delta(v_k) \mid \clE]&=\sum_{k=1}^{10d} \E[\Delta(u_k) \mid \clE] \nonumber \\
&\geq \Pr[\mbox{$\q_1$ is set to to $1$ within first $10d$ iterations} \mid \clE] \mbox{\ \ \ \ (by union bound)} \nonumber \\
&\geq \Pr[\mbox{$\q_1$ is set to to $1$ within first $10d$ iterations}]-\Pr[\overline{\clE}] \nonumber \\
& \geq \frac{9}{10}-\frac{1}{4}=\frac{13}{20}. \nonumber
\end{align}
\end{proof}
The following observation will be useful in the proof of Claim~\ref{sumofdelta2}.
\begin{observation}
\label{recursive}
Let $v$ be any node of $\clB$, such that the associated subcube has non-empty intersections with the supports of both $\mu_0$ and $\mu_1$. Let $\mu_0':=\mu_0 \mid v$ and $\mu_1':=\mu_1 \mid v$. Let $\clB_v$ denote the subtree of $\clB$ rooted at $v$. Then $\clB_v$ is an optimal  decision tree for $\mu_0'$ and $\mu_1'$.
\end{observation}
\begin{proof}
If $\clB_v$ is not an optimal decision tree for $\mu_0'$ and $\mu_1'$ then we could replace it by an optimal decision tree for $\mu_0'$ and $\mu_1'$, and for the resultant decision tree $\clB'$, the expected value of $\clN_1$ in $\clP(\clB', \mu_0, \mu_1)$ will be smaller than that in $\clP(\clB, \mu_0, \mu_1)$. This will contradict the optimality of $\clB$.
\end{proof}
\begin{proof}[Proof of Claim~\ref{sumofdelta2}]
For $i=0,\ldots, c-1$, let $w$ be any vertex at depth $10id+1$ consistent with $\clE$, such that $\Pr_{x \sim \mu_0}[x \in w], \Pr_{x \sim \mu_1}[x \in w] \neq 0$. Consider the subtree $\mathsf{T}$ of $\clB$ rooted at $w$. Let $w_1:=w$ and $w_\ell$ be the random vertex at depth $\ell$ of $\mathsf{T}$, when $\mathsf{T}$ is run on a random input from $\mu \mid w$, or $\bot$ if $\mathsf{T}$ terminates before $\ell$ queries. By Observation~\ref{recursive}, $\mathsf{T}$ is an optimal decision tree for distributions $\mu_0':=\mu_0 \mid w, \mu_1':=\mu_1 \mid w$. From Claim~\ref{sumofdelta} we have that,
\begin{align}
\sum_{\ell=1}^{10d} \E[\Delta(w_\ell) \mid \clE]  \geq \frac{13}{20}. \label{onetree}
\end{align}
Where $\Delta(w_\ell)$ is with respect to distributions $\mu'_0$ and $\mu'_1$. Since $\mu'_0 \mid w_\ell=\mu_0 \mid w_\ell$ and $\mu'_1 \mid w_\ell=\mu_1 \mid w_\ell$, $\Delta(w_\ell)$ in (\ref{onetree}) is also with respect to distributions $\mu_0$ and $\mu_1$. Now, when $w$ is the random vertex $v_{10id+1}$, $w_\ell$ is the random vertex $v_{10id+\ell}$. Thus from (\ref{onetree}) we have that,
\begin{align}
\sum _{k=10id+1}^{10(i+1)d} \E[\Delta(v_k) \mid \clE]  \geq \frac{13}{20}. \label{oneslab}
\end{align}
The claim follows by adding (\ref{oneslab}) over $i=0, \ldots, c-1$.
\end{proof}
Now we are ready to prove Claim~\ref{keyinf}. By setting $c=d$ and invoking Claim~\ref{sumofdelta2} we have,
\begin{align}
\sum_{t=1}^{10d^2} \E[\Delta(v_t) \mid \clE] \geq \frac{13d}{20}. \label{deltabound}
\end{align}
Let $\mathsf{E}_j$ be the event $\Pr_{x \sim \mu_0}[x \in v_j] \neq 0 \wedge \Pr_{x \sim \mu_1}[x \in v_j] \neq 0 \wedge v_j \neq \bot$ (i.e. $v_j$ is a vertex of $\clB$ and is not a leaf). Note that if $v_j \neq \bot$ and $v_j$ is not a leaf of $\clB$, $v_j$ is determined by $(a_1, \ldots, a_{j-1})$ and vice versa, and hence $\I(a_j:g(x) \mid a_1, \ldots, a_{j-1})=\I(x_{i_j}:g(x) \mid v_j)$. If $v_j=\bot$ or $v_j$ is a leaf of $\clB$, then $g(x)$ is determined by $(a_1, \ldots, a_{i-1})$, and $x_{i_j}=\bot$; thus, $\I(a_j:g(x) \mid a_1, \ldots, a_{j-1})=\I(x_{i_j}:g(x) \mid v_j)=0$. Thus we have,
\begin{align}
&\I(a_1, \ldots, a_{10d^2}:g(x)) \nonumber \\
&= \sum_{j=1}^{10d^2} \I(a_j:g(x) \mid a_1, \ldots, a_{j-1}) \mbox{\ \ \ \ (By the chain rule of mutual information (Theorem~\ref{chainrule}))}\nonumber \\
&= \sum_{j=1}^{10d^2} \I(x_{i_j}:g(x) \mid v_j) \mbox{\ \ \ \ \ \ \ \ \ \ \ \ \ \ \ \ (From the discussion above)} \nonumber \\
& \geq 8 \sum_{j=1}^{10d^2} \E \lf[\mathbf{1}_{\mathsf{E}_j}\cdot\lf[\Pr[g(x)=0 \mid x \in v_j] \cdot \Pr[g(x)=1 \mid x \in v_j] \cdot \Delta(v_j)\rt]^2\rt] \nonumber \\
&\qquad \qquad \qquad \qquad \mbox{\ \ \ \ \ \ \ \ \ \ \ \ \ \ \ \ \ \ \ \ \ \  (From Claim~\ref{mutin})} \nonumber \\
&\geq 8 \sum_{j=1}^{10d^2} \Pr[\clE] \cdot \E\lf[\lf[\Pr[g(x)=0 \mid x \in v_{j}] \cdot \Pr[g(x)=1 \mid x \in v_{j}] \cdot \Delta(v_j)\rt]^2 \mid \clE\rt] \nonumber \\
&\qquad \qquad \qquad \qquad \mbox{\ \ \ \ \ \ \ \ \ \ \ \ \ \ \ \ \ \ \ \ \ \ (Conditioned on $\clE, \mathsf{E}_j$ happens with probability $1$ for each $j \leq 10d^2$)} \nonumber \\
&\geq 8 \sum_{j=1}^{10d^2} \frac{3}{4} \cdot \frac{1}{9} \cdot \E[{\Delta(v_j)}^2 \mid \clE] \mbox{\ \ \ \ \ \  \ \  (By the assumption $\Pr[\clE] \geq \frac{3}{4}$ )} \nonumber \\
&= \frac{2}{3}\sum_{j=1}^{10d^2}   \E[{\Delta(v_j)}^2 \mid \clE]  \nonumber \\
&\geq \frac{2}{3}\sum_{j=1}^{10d^2}   \lf(\E[{\Delta(v_j)} \mid \clE]\rt)^2 \mbox{\ \ \ \ \ \ \ \ \ \ \ \ \ \ \ \ (By Jensen's inequality)} \nonumber \\
&\geq \frac{2}{3} \cdot \frac{1}{10d^2} \lf(\sum_{j=1}^{10d^2}  \E[\Delta(v_j) \mid \clE]\rt)^2 \mbox{\ \ \ \ (By Cauchy-Schwarz inequality)} \nonumber \\
&\geq \frac{1}{40}.\mbox{\ \ \ \ \ \ \ \ \ \ \ \ \ \ \ \ \ \ \ \ \ \ \ \ \ \ \ \ \ \ \ \ \ \ \  \ \ \ (From~(\ref{deltabound}))}\label{infbound}
\end{align}
\section{Proof of Claim~\ref{matchingex}}
\label{matchingexample}
\begin{proof}
\begin{enumerate}
\item[(i)] Assume that a deterministic protocol of cost $k$ solves $f_0$ with respect to the uniform input distribution with error at most $\dr49$.
Such a protocol partitions $\01^n$ into (at most) $2^k$ sub-cubes, each marked by some ``answer'' (an element from $\01^n$).
In particular, more than $2^n-2^{n-4}$ points belong to sub-cubes of size at least $2^{n-k-4}$ -- in other words, to sub-cubes of co-dimension at most $k+4$.
As more than \dr{15}{16} fraction of all points belong to such sub-cubes and the total protocol error is at most $\dr49$, there exists at least one sub-cube of co-dimension $k+4$, on which the protocol errs with probability less than $\frac{4}{9} \cdot \frac{16}{15} < \frac12$.

The symmetry in the definition of $f_0$ allows us to assume without loss of generality that the sub-cube is the set $\tau\deq0^{k+4}\circ\01^{n-k-4}$, where ``$\circ$'' denotes string concatenation.
It is easy to see that the ``answer'' that would minimize the error probability with respect to this sub-cube can be any binary string starting with ``$0^{k+4}$'', so let us assume that the actual label is $0^n$. Let $\mathcal{U}_{\ell}$ denote uniform distribution on $\01^n$.
Then
\m{
  \PRr{\tit{error}}{Z\in\tau}
  =\PRdg[Z'\sim\mathcal{U}_{n-k-4}]{\sz{Z'}\le\fr n2-\sq n}
  <\fr12
,}
which implies that $k+4\ge2\sq n$, as a uniformly-random binary string of length more than $n-2\sq n$ would have more than $\dr n2-\sq n$ ``ones'' with probability at least $\dr12$.
\item[(ii)] A randomized query protocol of cost $k$ and error $\dr13$ for $g_0$ would trivially imply existence of a randomized communication protocol of cost at most $2k$ and error $\dr13$ for the bipartite problem \e{Gap-Hamming-Distance}:
\m{
  GHD(X,Y)\deq\thrcase
    {0}{if $\sz{X\xor Y}\le\dr n2-\sq n$;}
    {1}{if $\sz{X\xor Y}\ge\dr n2+\sq n$;}
    {*}{otherwise,}
}
and it has been demonstrated by Chakrabarti and Regev~\cite{CR11_An_O} that the complexity of this problem for any constant error is \asOm n.
\item[(iii)] Consider the following protocol for computing $f_0\l(g_0(x_1)\dc g_0(x_n)\r)$, where $x_i\in\01^n$:
For every $i\in[n]$, let $a_i=x_i(j_i)$, where $j_i$ is chosen uniformly at random from $\{1, \ldots,n\}$ -- that is, $a_i$ is a uniformly-random bit of $x_i$.
Then $\sszz{i}{a_i=g_0(x_i)}$ -- the expected number of ``correctly guessed'' \pl[a_i] is at least $\dr n2+\sq n$; intuitively, this means that the probability that $a_1\dc a_n$ is a right answer to $f_0\l(g_0(x_1)\dc g_0(x_n)\r)$ is ``non-trivially high'' -- to ``boost'' this probability, we will use several ``probes'' from every $x_i$ and take their majority vote.

\noindent
\bil{Protocol:}
For an odd integer $t_\eps$ as defined next, independently choose $j_{i,k}\unin[n]$ for $i\in[n]$ and $k\in[t_\eps]$.
Let $a_i\deq maj\l(x_i(j_{i,1})\dc x_i(j_{i,t_\eps})\r)$ and output ``$a_1\dc a_n$''.

To analyse it, we consider for every $i\in[n]$:
\m{
 &\PRdg{a_i=g_0(x_i)}-\PRdg{a_i\ne g_0(x_i)}\\
  &\tb\ge\sum_{i=0}^{\fr{t_\eps-1}2}
   \chs{t_\eps}i\tm
   \l(
    \l(\fr12-\fr1{\sq n}\r)^i \l(\fr12+\fr1{\sq n}\r)^{t_\eps-i}
     -\l(\fr12-\fr1{\sq n}\r)^{\fr{t_\eps+1}2+i}
       \l(\fr12+\fr1{\sq n}\r)^{\fr{t_\eps-1}2-i}
   \r)\\
  &\tb=
    \l(1-\l(\fr{1-\dr2{\sq n}}{1+\dr2{\sq n}}\r)^{\fr{t_\eps+1}2}\r)
   \tm\sum_{i=0}^{\fr{t_\eps-1}2}
      \chs{t_\eps}i\tm
      \l(\fr12-\fr1{\sq n}\r)^i \l(\fr12+\fr1{\sq n}\r)^{t_\eps-i}
,}
where the equality occurs when $\sz{x_i}-\dr n2=\pm\sq n$.
As
\m{
  \sum_{i=0}^{\fr{t_\eps-1}2}
     \chs{t_\eps}i\tm
     \l(\fr12-\fr1{\sq n}\r)^i \l(\fr12+\fr1{\sq n}\r)^{t_\eps-i}
   =\PRr{a_i=g_0(x_i)}{\sz{x_i}-\fr n2=\pm\sq n}
   >\fr12
,}
we get
\m{
 &\PRdg{a_i=g_0(x_i)}-\PRdg{a_i\ne g_0(x_i)}\\
  &\tb>\fr12\tm
    \l(1-\l(\fr{1-\dr2{\sq n}}{1+\dr2{\sq n}}\r)^{\fr{t_\eps+1}2}\r)
   >\fr12\tm\l(1-\l(1-\fr2{\sq n}\r)^{\dr{t_\eps}2}\r)
   \ge\Min{\fr{t_\eps}{4\sq n},\fr14}
.}
Our $t_\eps$ will be small enough to guarantee that $\fr{t_\eps}{4\sq n}\le\fr14$, so we can write
\m[m_aig0]{
  \PRdg{a_i=g_0(x_i)} > \fr12 + \fr{t_\eps}{8\sq n}
.}

Now let us estimate the probability that $a_1\dc a_n$ is a wrong answer to $f_0\l(g_0(x_1)\dc g_0(x_n)\r)$:
This occurs only if $\sszz{i}{a_i=g_0(x_i)}<\dr n2+\sq n$, so by the Chernoff bound (in a form given in~\cite{DM05_Con_Bo}),
\m{
  \PRdg{\txt{the protocol errs}}
   <\exp\l(-\fr12\tm\l(\fr{t_\eps}8-1\r)^2\r)
,}
so that choosing $t_\eps\in\asT{\sq{\log(\dr1\eps)}}$ would suffice for our needs and the result follows.
\end{enumerate}
\end{proof}
\end{appendix}
\newcommand{\etalchar}[1]{$^{#1}$}

\end{document}